\documentclass{article}
\usepackage{graphicx} 
\usepackage{xcolor}
\usepackage[margin=1in]{geometry}
\usepackage{amsmath, amsthm, amssymb}
\usepackage{hyperref}
\hypersetup{
    colorlinks,
    linkcolor={red!80!black},
    citecolor={blue!80!black},
    urlcolor={blue!80!black}
}
\newtheorem{theorem}{Theorem}[section]
\newtheorem{lemma}[theorem]{Lemma}
\newtheorem{corollary}[theorem]{Corollary}

\newtheorem{observation}[theorem]{Observation}
\newtheorem{example}[theorem]{Example}
\newtheorem{claim}[theorem]{Claim}
\newtheorem{proposition}[theorem]{Proposition}
\newtheorem{definition}[theorem]{Definition}
\newtheorem{maintheorem}{Main Theorem}

\definecolor{awgreen}{RGB}{0, 150, 0}

\usepackage{graphicx, nicefrac}

\usepackage[]{color-edits}

\addauthor{MF}{blue}

\addauthor{TE}{red}

\addauthor{PD}{magenta}

\newcommand{\opt}{OPT}

\addauthor{TK}{brown}

\newcommand{\ind }[1]{\bold{1}[#1]}

\makeatletter
\def\moverlay{\mathpalette\mov@rlay}
\def\mov@rlay#1#2{\leavevmode\vtop{%
   \baselineskip\z@skip \lineskiplimit-\maxdimen
   \ialign{\hfil$\m@th#1##$\hfil\cr#2\crcr}}}
\newcommand{\charfusion}[3][\mathord]{
    #1{\ifx#1\mathop\vphantom{#2}\fi
        \mathpalette\mov@rlay{#2\cr#3}
      }
    \ifx#1\mathop\expandafter\displaylimits\fi}
\makeatother

\newcommand{\bigcupdot}{\charfusion[\mathop]{\bigcup}{\cdot}}
\newcommand{\nash}[1]{\mathsf{NE}(#1)}

\newcommand{\agents}{[n]}
\newcommand{\actions}{A}
\newcommand{\actionset}{A}
\newcommand{\reward}{r}
\newcommand{\outcomes}{\Omega}
\newcommand{\cost}{c}
\newcommand{\reals}{\mathbb{R}}
\newcommand{\contract}{\vec{\alpha}}

\newcommand{\Seq}{S^\dagger}

\usepackage{algorithm}

\usepackage{algpseudocode}

\title{Multi-Agent Combinatorial Contracts\thanks{Tomer Ezra is supported by the Harvard University Center of Mathematical Sciences and Applications. The work of M. Feldman has been partially funded by the European Research Council (ERC) under the European Union's Horizon 2020 research and innovation program (grant agreement No. 866132), by an Amazon Research Award, by the NSF-BSF (grant number 2020788), and by a grant from TAU Center for AI and Data Science (TAD). 
}}
\author{Paul D\"utting\thanks{Google Research, Zurich, Switzerland. Email:  \texttt{duetting@google.com}} \and Tomer Ezra\thanks{Harvard University, USA. Email: \texttt{tomer@cmsa.fas.harvard.edu}} \and Michal Feldman\thanks{Tel Aviv University, Israel. Email: \texttt{mfeldman@tauex.tau.ac.il}} \and Thomas Kesselheim\thanks{University of Bonn, Germany. Email: \texttt{thomas.kesselheim@uni-bonn.de}.}}
\date{}

\begin{document}

\maketitle

\begin{abstract}
Combinatorial contracts are emerging as a key paradigm in algorithmic contract design, paralleling the role of  combinatorial auctions in algorithmic mechanism design. 
In this paper we study natural combinatorial contract settings involving teams of agents, each capable of performing multiple actions. 
This scenario extends two fundamental special cases previously examined in the literature, namely the single-agent combinatorial action  model of \cite{DuttingEFK21} and the multi-agent binary-action model of \cite{BabaioffFNW12,DuettingEFK23}. 

We study the algorithmic and computational aspects of these settings, highlighting the unique challenges posed by the absence of certain monotonicity properties essential for analyzing the previous special cases. 
To navigate these complexities, we introduce a broad set of novel tools that deepen our understanding of combinatorial contracts environments and yield good approximation guarantees. 

Our main result is a constant-factor approximation for submodular multi-agent multi-action problems with value and demand oracles access. This result is tight: we show that this problem admits no PTAS (even under binary actions). 
As a side product of our main result, we devise an FPTAS, with value and demand oracles, for single-agent combinatorial action scenarios with general reward 
functions, which is of independent interest.
We also provide bounds on the gap between the optimal welfare and the principal's utility. 
We show that, for subadditive rewards, perhaps surprisingly, this gap scales only logarithmically (rather than linearly) in the size of the action space.
\end{abstract}

\pagenumbering{arabic} 

\section{Introduction}

A fundamental problem in economics is to incentivize a team of agents to exert effort in a joint project (see, e.g., \cite{Holmstrom82}). 
Often the project is owned by a principal who can only observe the final outcome (e.g., whether the project succeeds or not), but not the combination of actions chosen by each agent. 
Since the principal cannot directly observe the agents' choice of actions, she incentivizes the agents by a contract that specifies outcome-contingent payments.
The two intertwined questions are: what is the desired set of actions and how should the principal incentivize the agents to take them?

\paragraph{Model.}
A possible abstraction of this problem is as follows. A principal interacts with a set of $n$ agents. Each agent can take any subset $S_i$ of actions $\actions_i$. We write $\actions = \bigcupdot \actions_i$ for the set of actions, and $S = \bigcupdot S_i$ for a choice of actions by the agents. Note that we may have $S_i = \emptyset$, in which case agent $i$ does not take any action. Each action $j \in \actions$ comes with a cost $c_j$ for the respective agent. Costs are additive, so that $c(S) = \sum_{j \in S} c_j$. In particular, $c(\emptyset) = 0$, so that not taking any actions comes with no cost.
There are two possible outcomes: success or failure of the project. The principal's reward in case of success is $\reward$, and it is zero otherwise.  
Each set of actions $S$ induces a success probability $f(S)$ of the entire project. The success probability $f$ is normalized, so that $f(\emptyset) = 0$, and monotonically non-decreasing, so that $f(S) \leq f(T)$ for any $S \subseteq T$.  
The principal cannot observe the agents' actions, only the final outcome. The principal incentivizes the agents through a (linear) contract 
$\contract = (\alpha_1, \ldots, \alpha_n)$, which specifies the fraction of $\reward$ to be paid to each agent $i$ in case of success. Consider a choice of actions $S$, then agent $i$ best responds to contract $\contract$ if the set of actions $S_i = S \cap \actions_i$ chosen by agent $i$ maximizes agent $i$'s expected utility given by $\alpha_i f(S) \cdot \reward - c(S_i)$. Note that agent $i$'s utility depends on the choice of actions $S_{-i} = S \setminus S_i$ of the agents other than $i$. A profile $S$ in which all agents $i \in \agents$ best respond, is a (pure) Nash equilibrium.
The principal's goal is to maximize her expected profit, given by $(1-\sum_i \alpha_i) f(S) \cdot \reward$, when agents choose a set of actions $S$ that is a best response to the announced contract. 

\paragraph{Related previous work.}
This contracting model is the joint generalization of two fundamental combinatorial contracting models. The first model is the single-agent combinatorial action model of \cite{DuttingEFK21}, which is the special case of the above model with $n = 1$ agent. The second model is the multi-agent binary action model of \cite{BabaioffFNW12,DuettingEFK23}, which corresponds to the special case of our model where each agent controls a single action, say $A_i = \{i\}$ for each agent $i \in \agents$. 
However, the approaches and techniques that have been developed by prior work for the respective models  
seem rather different and incompatible. 

Indeed, in the single-agent combinatorial action model a contract is just a scalar $\alpha$, and the problem has ``linear structure.'' Specifically, one can look at the agent's utility for each set of actions $S \subseteq \actions$ as a function of $\alpha \in [0,1]$. This gives rise to an upper envelope diagram, with breakpoints where the best choice of actions changes. These breakpoints are referred to as critical $\alpha$'s, and it's clear that an optimal contract must occur at a critical $\alpha$. Known poly-time algorithms for this problem exploit this structure by enumerating all critical values, which in particular requires that there are at most poly-many critical values. In the realm of complement-free $f$, this approach is  limited to gross-substitutes (GS) $f$, because for submodular $f$ there can be exponentially many breakpoints \cite{DuttingEFK21}.

In the multi-agent binary actions case, in contrast, for any subset $S$ of agents $A$ it's clear what the corresponding $\alpha_i$ should be. Indeed, for each agent $i \in S$ we can compare this agent's utility under $S$, namely $\alpha_i f(S) \cdot \reward - c_i$, to that if only the agents in $S \setminus \{i\}$ exert effort, namely $\alpha_i f(S \setminus \{i\}) \cdot \reward$, and solve for the smallest $\alpha_i$ that makes the agent indifferent. For all agents $i \not\in S$, we can set $\alpha_i = 0$. The difficulty in this case is that ``nice structure'' of $f$ (say submodularity) does not translate into nice structure of the principal's objective function $(1-\sum_i \alpha_i(S)) f(S) \cdot \reward$. The problem of maximizing the principal's utility admits a poly-time constant-factor approximation for submodular (and even XOS) $f$, with value (and demand) oracle access to $f$ \cite{DuettingEFK23}. The approach taken by this algorithm is to search for a ``good'' set of agents $S$ by first issuing a demand query to $f$ at suitably chosen ``item prices'' $p_i$, one for each agent $i \in \agents$. Afterwards, the algorithm prunes this set $S$ by dropping agents (using that if an agent is dropped, other agents only want to do more) to reach a state where all the remaining agents have high marginal contributions to $f$.
High marginals are important as they correspond to small payments (in terms of $\alpha_i$).

\paragraph{Challenges in the combined problem.}
The multi-agent multi-action problem that we study in this work loses the linear nature of the single-agent combinatorial action problem, and it's also unclear how to extend the multi-agent binary actions approach. 

Specifically, when an agent controls more than one action, the analogy to item prices underlying the existing approaches to the multi-agent binary-actions setting breaks. In fact, unlike in the binary actions case, there is a rather non-trivial equilibrium selection (and construction) problem. 
First and foremost, there is no longer a simple formula for the best contract $\contract$ that incentivizes a given set $S$. In particular, while in the binary actions case, we only have one-sided constraints on the $\alpha_i$'s, in the multi-action case we have upper and lower bounds. As a consequence, not all sets $S$ can be incentivized, and even fixing $S$ finding the range of $\alpha_i$'s that incentivize $S_i$ may require checking exponentially many inequalities.

The intricate interaction between agents and the actions they choose has some important implications. In particular, it causes the problem to lose some monotonicity properties that have been essential in the analysis of the previously-studied special cases.  
Specifically, in the binary actions case, for submodular $f$, under any minimal contract for a set $S$ (where the $\alpha_i$'s are as small as possible), it is possible to drop agents from $S$ setting their $\alpha_i$'s to zero, and this will keep the equilibrium intact. So if other agents do less, then agents only want to do more. As a consequence, the principal's utility is subadditive in the set of agents who are incentivized. This is no longer the case in the multi-agent multi-action setting as the following example demonstrates.

\begin{example}[One agent might do less as a result of others doing less]
\label{example:more-or-less}
Consider a principal-agent setting with one principal and two agents. The principal's reward is $\reward = 1$. The action sets are $A_1=\{1,2\}$, $A_2=\{3\}$. The costs are 
$c_1=0.1$, $c_2=0.24$, $c_3=\epsilon$ for $\epsilon > 0$ small enough. 
The success probabilities are $f(\{1\})= \frac{7}{16}, f(\{2\}) = \frac{3}{4}, f(\{3\})=\frac{1}{4},f(\{1,3\})=\frac{9}{16}, f(\{1,2\})=\frac{7}{8},f(\{2,3\})=f(\{1,2,3\})=1$.  
Note that the success probability function defined this way is submodular but not gross substitutes.

\begin{itemize}
\item The principal utility is maximized by incentivizing $\{2,3\}$ with contract $\contract = (\frac{8}{25},4\epsilon)$, which gives the principal a utility of $1-\frac{8}{25}-4\epsilon=\frac{17}{25} - 4 \epsilon = 0.68-4\epsilon$.
\item The optimal way to incentivize $\{1\}$ is by using contract $\contract = (\frac{8}{35},0)$, which gives the principal a utility of $(1-\frac{8}{35})\frac{7}{16} = \frac{27}{80} = 0.3375$
\item The optimal way to incentivize $\{2\}$ is by using contract $\contract = (\frac{56}{125},0)$, which gives the principal a utility of $(1- \frac{56}{125})\frac{3}{4} = \frac{207}{500} = 0.414 $
\item The optimal way to incentivize $\{1,2\}$ is by using contract $\contract = (\frac{4}{5},0)$, which gives the principal a utility of $(1-\frac{4}{5}) \frac{7}{8} = \frac{7}{40} = 0.175$.
\item The optimal way to incentivize $\{3\}$ is by using contract $\contract = (0, 4\epsilon)$, which gives the principal a utility of $\frac{1}{4}(1- 4\epsilon) = \frac{1}{4}-\epsilon = 0.25-\epsilon$.
\end{itemize}
We note that it is also possible to incentivize $\emptyset$ and $\{1,3\}$ but these give smaller principal utility. 

We can first observe that when agent~2 is dropped, then we need to pay more to agent~1 to incentivize him to take action 2 ($\frac{56}{125} > \frac{8}{25}$). This in particular means that when we keep $\alpha_1 = \frac{8}{25}$ and set $\alpha_2 = 0$, then agent $2$ will take no action, while agent $1$ will do $\{1\}$. So one agent (agent 2) doing less causes another agent (agent 1) to do less as well.

We also observe, that the best contracts for agent $1$ individually and agent $2$ individually incentivize $\{2\}$ and $\{3\}$ respectively. The sum of principal's utilities of these contracts is $\frac{207}{500} + \frac{1}{4} - \epsilon = 0.664 - \epsilon$. This is less than the principal's utility in the best contract for agents $1$ and $2$, which is $\frac{17}{25} - 4 \epsilon = 0.68 - 4\epsilon$ for $\epsilon$ small enough.
\end{example}

\subsection{Our Results}

We start with a simple but important observation (which was independently observed in \cite{DeoCampoVuongEtAl2024} for the multi-agent binary actions model). Namely, for any fixed contract $\contract$, the induced game among the agents is a potential game (see Section~\ref{sec:model},  Proposition~\ref{prop:potential-function}). This shows that any contract $\contract$ admits at least one (pure) Nash equilibrium. 
Moreover, from the  particular form of the potential function, it's also clear that one can find \emph{some} equilibrium $S$ with a single demand query to $f$.

\paragraph{Constant-factor approximation.}  Our first main result is a polynomial-time algorithm that gives constant approximation to the optimal principal's utility for any setting with submodular success probability function $f$, under value and demand oracles.
Notably, our algorithm finds a contract $\contract$, so that \emph{any} equilibrium $S$ of this contracts yields a high utility for the principal. 
This worst-case type guarantee is significantly stronger than simply ensuring the existence of a good equilibrium.

\begin{maintheorem}[See Sections \ref{sec:no-agent-is-large}, \ref{sec:single-agent}, \ref{sec:putting-together}]\label{mainthm:1}
There exists a poly-time algorithm, issuing poly-many value and demand queries that, for every multi-agent multi-action contract setting with submodular reward function, computes a contract $\contract$ such that every equilibrium of $\contract$ gives a constant approximation to the optimal principal's utility. 
\end{maintheorem}

Our proof of this theorem 
requires several new insights, and new techniques. The proof
reduces the problem to one of two cases: 
Either no agent is large, or only a single agent is incentivized. 
(We remark that although a similar reduction appeared in \cite{DuettingEFK23}, it's all but obvious that this would also hold beyond binary actions, because of the non-monotonicity illustrated in Example~\ref{example:more-or-less}).

For the ``no agent is large'' case we give a constant-factor approximation with access to a value oracle.  
Our algorithm 
relies on a novel concept of subset stability, and a new Doubling Lemma, which links this concept to contracts and equilibria. 
Our approach for this case also uncovers 
an interesting 
connection to a particular form of bundle prices, that we exploit to argue that
we can efficiently optimize over these prices.

To address the ``only a single agent is incentivized'' case, we first show how to get an FPTAS for the single-agent problem with access to a value and demand oracle, even for general reward functions. The key difficulty in showing this result is to argue that the optimal contract is strictly bounded away from $1$. 
Our argument for this exploits the (known) upper bound on the utility-to-welfare gap \cite{DuttingRT19}. 
We then show how to extend and robustify the resulting contracts for the setting with multiple agents, using subset stability and the Doubling Lemma, losing only 
an additional constant factor.

The FPTAS that we give for the single-agent problem  
is of independent interest. It improves upon a weakly FPTAS that has been established by \cite{DuttingEFK21}.
It is also tight in that it is known that, even for submodular $f$ computing an optimal contract is \textsf{NP}-hard \cite{DuttingEFK21}, and may require exponentially many demand queries \cite{DuettingFGR24}.

We remark that Main Theorem~\ref{mainthm:1} also implies the first constant-factor approximation for the multi-agent multi-action problem with gross substitutes $f$, with value oracle access only. This problem is known to be $\mathsf{NP}$-hard,
even for additive $f$ and binary actions \cite{DuettingEFK23}.

\paragraph{Tightness/impossibility.} Our second main result shows that Main Theorem~\ref{mainthm:1} is essentially tight. Namely, it shows that constant approximation that it establishes for submodular $f$ with value and demand oracle access is the best we can hope for, even with binary actions.

\begin{maintheorem}[Section ~\ref{sec:noptas}]\label{mainthm:2}
No PTAS exists for the optimal contract problem in multi-agent  
contract settings with submodular reward and with value and demand oracles, even in the special case of binary actions.
\end{maintheorem}

This theorem strengthens two previous impossibility results for the multi-agent binary actions problem. Specifically, for submodular rewards $f$, it was known that there cannot be a PTAS, assuming value oracle access only \cite{EzraFS24}.
Under demand oracle access, it was only known that there can be no PTAS for the broader class of XOS success probability functions $f$ \cite{DuettingEFK23}. 

Our proof of the impossibility theorem differs quite significantly from these previous impossibilities: Rather than hiding a valuable set, it hides a set where the slopes (= marginals) are higher (and hence the payments are smaller).

\paragraph{Second best vs.~first best.}
The social welfare of an action set $S$ is the difference between its reward and cost, namely $f(S)-\cost(S)$. Our final set of results concerns the gap between the optimal principal's utility from a contract and the optimal social welfare of an instance.  
This gap is also known as second-best vs.~first-best in economics. 

\begin{maintheorem}[See 
Section \ref{sec:first-best}]
\label{mainthm:3}
The following guarantees hold with respect to the first-best:
\begin{itemize}
    \item For a single agent with  subadditive reward function $f$, the gap between the optimal social welfare and the best principal's utility is upper bounded by the number of actions, $m$. Moreover, the same holds for the best equilibrium in multi-agent settings with subadditive reward $f$.
    \item Combined with Main~Theorem~\ref{mainthm:1}, for submodular rewards $f$, the best contract obtains an $\Omega(\frac{1}{m})$ fraction of the optimal social welfare under every equilibrium.
\end{itemize}
\end{maintheorem}

Two remarks are in order. First, this result may be surprising in light of previous results. Specifically, \cite{DuttingRT19} show that the gap between the optimal social welfare and the best principal utility is bounded by the number of critical $\alpha$'s, while \cite{DuttingEFK21} showed that even for submodular reward functions the number of best responses can be exponential in $m$. 
Thus, our result shows that the actual bound is exponentially better than what one could expect given \cite{DuttingRT19,DuttingEFK21}.
Second, this result, beyond being interesting in its own right, implies an improved query complexity of our FPTAS for the single-agent combinatorial actions setting when rewards are subadditive.

\subsection{Our Techniques}\label{sec:techniques}

\paragraph{Techniques for Main Theorem~\ref{mainthm:1}.}
We first present and discuss the notion of subset stability (Definition~\ref{def:subset-deviations}) along with the  Doubling Lemma (Lemma~\ref{lemma:anyequilibrium}), which are crucial components of our approximation algorithm.
Recall that, in an equilibrium of a contract, no agent can increase their utility by deviating to any other choice of actions. We introduce a weaker condition called subset stability: No agent wants to deviate to any subset of the actions. This condition turns out to be extremely powerful because of our Doubling Lemma (Lemma~\ref{lemma:anyequilibrium}). This lemma asserts that if $S$ is subset stable with respect to a contract $\contract$, then any equilibrium $\Seq{}$ of a contract $2 \contract + \vec{\epsilon}$ fulfills $f(\Seq{}) \geq \frac{1}{2} f(S)$. Therefore, in a way, it is sufficient to compute a contract admitting a subset stable set of high value. 
The advantage of subset stability (compared to standard stability required in equilibria) is that it admits useful properties, such as closure under removal of actions, making it easier to ensure.

Our method for designing an approximation algorithm involves taking the better of two contracts.
The first contract we compute approximates the best contract in the case that no agent is ``large'', meaning that, for the optimal contract $\contract^\star$ and the optimal equilibrium $S^\star$ of contract $\contract^\star$, for each agent $i \in \agents$, $f(S^\star_i)$ is bounded away from $f(S^\star)$. In this case, we design an algorithm to compute a contract $\contract$ and a subset stable set $S$ such that $\sum_{i \in \agents} \alpha_i \leq \frac{1}{4}$ and $f(S) = \Omega(1) \cdot f(S^\star)$. Then our Doubling Lemma shows that $2 \contract + \vec{\epsilon}$ is a good approximation of the optimal contract. Furthermore, while subset stability is easier to ensure than standard stability, it still requires verifying exponentially many conditions. To mitigate this problem, we identify a sufficient condition for subset stability. 
Morally, it suffices to find a set that, for some $\gamma > 0$, maximizes 
\[
f(S) - \gamma \sum_{i \in \agents} \sqrt{\sum_{j \in S_i} c_i}.
\]
This can be interpreted as finding a demand set with respect to a particular type of (non-additive) bundle prices. For this latter problem, we devise an approximation algorithm.

The second contract we compute approximates the best contract in which only one agent is incentivized, meaning that, in the optimal contract $\contract^\star$,  $\alpha^\star_i = 0$ for all but one agent $i$. We first assume that there is only a single agent, which is the setting studied in \cite{DuttingEFK21}. We show that this problem admits an FPTAS even for general $f$; so far only algorithms running in weakly polynomial time were known. 
To this end we show that, in the optimal contract for this agent, $\alpha$ is bounded away from $0$ and $1$. 
The bounds we derive exploit (known) bounds on the gap between first best and second best, in particular the fact that the principal can achieve a utility of at least $\frac{1}{2^m}\max_{S'}(f(S') - c(S'))$.
We then discretize the range of possible $\alpha$'s, and determine the best contract in this polynomial-sized set. This leads to a query complexity of $O(m^2).$ (In Proposition~\ref{prop:welfare-single}, we sharpen the lower bound on the utility the principal can achieve with a linear contract when $f$ is subadditive to  $\frac{1}{m}\max_{S'}(f(S') - c(S'))$. This yields an improved query complexity of $O(m \log m).$)
Transitioning back to the multi-agent setting, other equilibria may emerge, involving additional agents (beyond the intended one) taking actions.
By appealing once more to the concept of subset stability and the Doubling Lemma, we show that the loss to the principal's utility due to these agents can be bounded.

The final step is to argue that the optimal principal utility is well approximated (to within a constant factor) by the two cases: no agent is large, and only a single agent is incentivized. This last step involves demonstrating that in the presence of a large agent in the optimal solution, incentivizing a single agent is sufficient. Due to the lack of monotonicity observed in Example~\ref{example:more-or-less}, this property is not as straightforward as it may appear. Instead, we use the facts that every equilibrium is also subset stable and that subset stability is preserved under removal of actions, and then apply our Doubling Lemma.

\paragraph{Techniques for Main Theorem~\ref{mainthm:2}.}
To establish Main Theorem~\ref{mainthm:2}, our impossibility for the multi-agent setting with binary actions, we construct a family of hard instances based on a carefully chosen function on two types of agents, ``good" agents which we want to incentivize in the optimal contract, and ``bad" agents which we don't want to exert effort. While this general approach was also adopted by former constructions \cite{DuettingEFK23,EzraFS24}, our approach differs from these prior approaches in a fundamental way. Namely, rather than hiding a set of high value, we hide a set with high marginals.

In our construction, 
the value of a set that contains a higher fraction of good agents is no higher than the value of a set of agents of the same size with a lower fraction of good agents. 
Moreover, the value is strictly lower, unless the fraction of the good agents is small enough, or the size of the set of agents is large enough. 
This construction makes the marginal contribution of every good agent to the set of good agents higher, which in turn decreases the payments necessary to incentivize them.
The good set of agents (and all its subsets) remains hidden, as demand oracles are good at finding sets with higher values as opposed to sets with lower values.

\paragraph{Techniques for Main Theorem~\ref{mainthm:3}.}
We first analyze the case of a single agent. For this case, instead of considering all critical $\alpha$'s, we split the interval of contracts $[0,1]$ to sub-intervals in places where new actions enter the best response of the agent for the first time. By the fact that the welfare can be written as the integral over the success probability as a function of the contract $\alpha$, and with the subadditivity of the success probability function, this split allows us to bound the social welfare as a sum of expressions, one for each action. 
On the other hand, the minimal contract that incentivizes an action $j$, is one of these points in the split, and achieves the corresponding expression in the bound of the welfare. This gives us the bound on the gap to the welfare of at most the number of actions. 

To extend it beyond the single agent setting, we design a randomized contract that incentivizes only a single agent. We choose the agent to incentivize with a probability that is proportional to the number of actions he has. This random choice balances the contributions of the agents since the approximation for the single agent is proportional to the inverse of the number of actions, which gives us the approximation of the number of actions. The optimal contract achieves at least the utility of the randomized contract we analyze.

\subsection{Related Work}

\paragraph{(Algorithmic) contract theory.} Contract theory is one of the pillars of microeconomic theory \cite{Ross73,Holmstrom79,GrossmanH83}, and has been recognized by the 2016 Nobel Prize to Hart and Holmstr\"om. 
Motivated by the fact that more an more of the classic applications of contract theory are moving online and growing in scale, recent work has started to explore contracts from an algorithmic perspective (see, e.g., \cite{BabaioffFNW12,HoSV14,DuttingRT19}).

A central topic in the emerging field of algorithmic contract theory are ``combinatorial contracts.'' The multi-agent binary action model is pioneered and explored in \cite{BabaioffFNW12,DuettingEFK23}, with \cite{DuettingEFK23} exploring the hierarchy of complement free success probability functions. The single-agent combinatorial action model was proposed and first studied in \cite{DuttingEFK21}. Additional work that provides both positive and negative results for these two models includes \cite{EzraFS24,DeoCampoVuongEtAl2024,DuettingFGR24,DuttingFG24}. As a strict generalization of both these models, our work inherits the hardness results that have been established for the two special cases.

Most closely related to our work is a recent working paper by \cite{cacciamani2024multi}. This paper studies a related but different multi-agent contracting setting with multiple actions per agent. This work, unlike our work, assumes that each agent $i$ has a fixed set of actions $\actions_i$, and can take any action $j \in \actions_i$. Moreover, any action profile (a choice of a single action for each agent) induces a probability distribution over possible outcomes (not necessarily binary). A main contribution of this work is the introduction of  the concept of randomized contracts for this setting (also see work on typed contracts below). 
They show that randomized contracts can be arbitrarily better than deterministic ones, and give an algorithm that finds a $(1+\epsilon)$-approximate randomized contract in time polynomial in $O(1/\epsilon)$ and the description size of the problem. The work also shows a reduction from the multi-agent setting to a single-agent setting (via virtual costs), and uses this to show an approximation guarantee for linear contracts.

Two additional directions in combinatorial contracts include \cite{DuttingRT21} and \cite{CastiglioniM023}. The former studies a single agent contracting setting with combinatorial outcome space \cite{DuttingRT21}, the latter studies a related (but different) multi-agent contracting problem, where each agent's action leads to an observable individual outcome which an agent's payment can depend on \cite{CastiglioniM023}.

The studies of \cite{DuttingRT19,BalmacedaEtAl16} explore the gap between linear contracts and optimal contracts, as well as between the best contract and the optimal social welfare. There is a significant amount of work on typed contracts, 
including \cite{GuruganeshSW21,CastiglioniM021,CastiglioniM022,AlonDT21,AlonDLT23}. There is also work that considers the problem of learning contracts, especially from an online learning, no-regret perspective \cite{HoSV14,DuettingGSW23,ZhuBYWJJ23}.

\paragraph{Combinatorial optimization.} On a technical level, our work is related to work on combinatorial optimization, in particular work on combinatorial auctions. Here, gross substitues appear as a natural frontier for economic but also algorithmic reasons (e.g., \cite{PaesLeme17}). In a seminal paper, \cite{Feige09}  provides a constant-factor approximation for the welfare maximization problem in combinatorial auctions with submodular, XOS, and subadditive bidders. An exciting line of work, seeks truthful approximation algorithms for submodular and XOS bidders \cite{Dobzinski21,AssadiKS21}. A different line of work has approached this problem, by analysing the Price of Anarchy of simple combinatorial auctions, such as combinatorial auctions with item bidding \cite{ChristodoulouKS16,FeldmanFGL13}. 
Related set-valued optimization problems come up in prophet inequalities and online posted price mechanisms \cite{FeldmanGL15,DuettingFGL17,DuttingKL20,CorreaC23}. There are also polynomial-time constant-factor approximation results for
truthful revenue maximization with unit-demand bidders \cite{ChawlaHMS10}, additive bidders
\cite{Yao15}, and XOS bidders \cite{CaiOZ22}. An important difference between these works on combinatorial allocation problems and our work, in addition to the different incentive problems arising due to the hidden action aspect of the problem, is that in our case the set-valued function $f$ is not additively separable across agents.

\paragraph{Optimizing the effort of others.} 
More broadly,
our work contributes to an emerging frontier in computer science, seeking to optimize the effort of others. Additional domains within this frontier encompass strategic classification (e.g., \cite{KleinbergR19,BechavodEtAl22}), optimal scoring
rule design (e.g., \cite{HartlineSLW23,YiHSW22}), and delegation (e.g., \cite{KleinbergK18,BechtelDP22}).

\section{Model}\label{sec:model}
\paragraph{Principal-agent model.} We consider the case where a single principal (she, her) delegates the execution of a project to $n$ agents (he, his). We focus on the binary outcome case, where the project can either succeed or fail. We refer to these two outcomes as $0$ and $1$, respectively, and let $\outcomes = \{0,1\}$. The principal receives a reward of $\reward > 0$ for outcome $1$, and a reward of zero for outcome $0$. A simple scaling argument shows that it is without loss of generality to assume that $r = 1$, and this would be our default assumption. 
However, for some examples and constructions it will be convenient to work with general $r > 0$. Each agent $i \in \agents$ has a set of actions $\actions_i$ that he can choose from. We assume that the sets of actions are disjoint, i.e., $\actions_i \cap \actions_{i'} = \emptyset$ for any pair of agents $i \neq i'$ with $i,i' \in \agents$. We write $\actionset = \bigcupdot_{i \in \agents} \actions_i$ for the set of actions. We use $m = |A|$ to refer to the total number of actions.

 Agent $i$ can take any subset $S_i \subseteq \actions_i$ of actions (including $S_i = \emptyset$). We write $S = \bigcupdot_{i \in [n]} S_i \subseteq \actionset$ for the set of actions taken by the agents. Each action $j \in \actionset$ comes with a (non-negative, real) cost $c_j \in \reals_+$.  
 The cost of a set of actions $S \subseteq \actionset$ is additive, i.e., $c(S) = \sum_{j \in S} c_j$ (with $c(\emptyset) = 0$). A success probability function $f: 2^A \rightarrow [0,1]$ maps each set of actions to a success probability.  When the reward for success is normalized to $r = 1$, the success probability coincides with the expected reward of a set of actions. We thus also refer to $f$ as the reward function.  

\paragraph{Contracts and utilities.} In general, a contract defines outcome-contingent transfers from the principal to the agent. In the binary outcome case that we focus on, it is without loss of generality to pay the agents only on success (see Observation II.1 in \cite{DuttingEFK21}).
It is therefore without loss of generality to consider linear contracts. A linear contract $\contract = (\alpha_1, \ldots, \alpha_n)$ specifies, for each agent $i \in \agents$, a fraction $\alpha_i \in [0,1]$ of the reward that should be transferred to the agent upon success. Suppose given contract $\contract$ the agents choose actions $S$, then the principal's utility is given by
\[
u_p(S,\contract) = \left(1-\sum_{i \in \agents} \alpha_i\right) \cdot f(S) \cdot r. 
\]
In turn, agent $i$'s utility is given by
\[
u_i(S,\contract) = \alpha_i \cdot f(S) \cdot r - \sum_{j \in S \cap \actions_i} c_j.
\]

Note that agent $i$'s utility depends on his own actions $S_i = S \cap \actions_i$, but also on the actions $S_{-i} = S \cap (\bigcup_{i' \neq i} \actions_{i'})$ chosen by the other agents. To highlight which actions are controlled by which agents, we sometimes write $f(S_i,S_{-i})$ to mean $f(S_i \cup S_{-i})$, and similarly for other functions. 
A choice of actions $S$ is a (pure) Nash equilibrium for contract $\contract$, if for each agent $i \in [n]$ and all $S'_i \subseteq \actions_i$ it holds that
\[
u_i(S_i,S_{-i},\contract) 
\geq u_i(S'_i,S_{-i},\contract).
\]

We note that the special structure of the games we consider ensures the existence of a pure Nash equilibrium for any $\contract$.
This can be shown by observing that $\phi(S) = f(S) - \sum_{i \in \agents} c(S_i)/\alpha_i$ is a potential function. (If $\alpha_i = 0$ we define $c_j/\alpha_i = \infty$ when $c_j > 0$, and we let $c_j/\alpha_i = 0$ when $c_j = 0$.) For the special case of binary actions, this  observation has been made by \cite{DeoCampoVuongEtAl2024}. A proof appears in Appendix~\ref{apx:potential-function}. 

\begin{proposition}[See also \cite{DeoCampoVuongEtAl2024}]\label{prop:potential-function}
    Every contract $\contract$ admits at least one pure Nash equilibrium $S$.
\end{proposition}

\paragraph{The optimization problem.} Our goal is to compute a contract $\contract$ and a choice of actions $S$ (to be recommended to the agents) such that the choice of actions $S$ constitutes a Nash equilibrium for contract $\alpha$ and such that 
\begin{equation} 
    u_p(S,\contract) \geq \kappa \cdot \sup_{\contract'} \sup_{S' \in \nash{\contract'}} u_p(S',\contract') = \kappa \cdot u_p(S^\star,\contract^\star), \nonumber\label{eq:principal-utility}
\end{equation}
where $\nash{\contract'}$ constitutes the set of actions that establish a Nash equilibrium between the agents under contract $\contract'$. We note that the supremum in the last equation is attained at some optimal contract $\contract^\star$ with corresponding equilibrium $S^\star$ (see Proposition~\ref{prop:sup-max}). We say that a pair of contract and actions is optimal if $\kappa = 1$. Otherwise, $\kappa \in [0,1)$ captures how close to optimal the solution is.

We remark that our algorithmic results are actually stronger: We find a contract $\contract$ such that any choice of actions $S'$ that constitutes an equilibrium for this $\contract$ achieves the claimed approximation guarantee. Namely, we find a contract $\contract$ such that
\[
\inf_{S' \in \nash{\contract}} 
u_p(S', \contract) 
\geq \kappa \cdot u_p(S^\star,\contract^*).
\]

\paragraph{Classes of reward functions.} 
We generally assume that the reward function $f$ is normalized so that $f(\emptyset) = 0$, and that it is non-decreasing, i.e., that for any $S \subseteq S' \subseteq A$ it holds that $f(S) \leq f(S').$ We sometimes use the shorthand $f(j) = f(\{j\})$. For an action $j$ and a set of actions $S \subseteq A$, we write $f(j \mid S) = f(S \cup \{j\}) - f(S)$ for the marginal contribution of $j$ to $S$. Our main interest will be in the following classes of reward functions $f: 2^A \rightarrow [0,1]$ that come from the hierarchy of complement-free functions \cite{LehmannLN06}: 
\begin{itemize}
\item Set function $f$  
is \emph{additive} if there exist values $v_1, \ldots, v_m \in \reals_+$ such that $f(S) = \sum_{j \in S} v_j$.
\item A set function $f$ is \emph{gross substitutes} if for any two vectors $p, q \in \reals_+^m$ and any $S \subseteq \actionset$ such that $f(S) - \sum_{j \in S} p_j \in \arg\max_{S' \subseteq \actionset} (f(S') - \sum_{j \in S'} p_j)$ there is a $T \subseteq \actionset$ such that $f(T) - \sum_{j \in T} q_j \in \arg\max_{T' \subseteq \actionset} (f(T') - \sum_{j \in T'} q_j)$ and $T \supseteq \{j \in S \mid q_j \leq p_j\}$.
\item Set function $f$ is \emph{submodular} if for every two sets $S, S' \subseteq \actionset$ with $S \subseteq S'$ and any $j \in \actionset$ it holds that $f(j \mid S) \geq f(j \mid S')$.
\item Set function $f$ is \emph{XOS} if there exist a collection of additive functions $\{a_s: 2^\actionset \rightarrow \reals_+\}_{\ell = 1,\ldots, k}$ such that for each $S \subseteq \actionset$ it holds that $f(S) = \max_{\ell=1,\ldots,k} \sum_{j \in S} a_{\ell j}$. Given an XOS function $f$ and a set $S$ there exists an additive function $a_\ell$ such that $a_\ell(S) = f(S)$ and $a_\ell(T) \leq f(T)$ for all $T \subseteq \actionset$; this function is called the \emph{additive supporting function} of $f$ on $S$.
\item Set function $f$ is \emph{subadditive} if for any two sets $S,S' \subseteq \actionset$, it holds that $f(S) + f(S') \geq f(S \cup S').$
\end{itemize}

A well-known fact is that additive $\subseteq$ gross substitutes $\subseteq$ submodular $\subseteq$ XOS $\subseteq$ subadditive, and that all containment relations are strict \cite{LehmannLN06}.

We emphasize that we do \emph{not} assume additive separability of $f$ across agents, i.e., we do \emph{not} assume that $f(S) = \sum_{i} f_i(S_i)$. While additive separability is a natural assumption in combinatorial auctions (where notions such as the welfare of an allocation are defined as the sum of the agents' valuations), it would preclude natural effects between agents and the actions they take in the contracts setting.

\paragraph{Value and demand oracles.} As is common in the literature on combinatorial optimization problems with set functions, we consider two types of primitives for accessing the reward function $f$.
\begin{itemize}
\item \emph{Value oracle access}: A value oracle for $f$ is given a set $S \subseteq \actionset$ and returns the value $f(S)$.
\item \emph{Demand oracle access:} A demand oracle is given $f$ and a (price) vector $p \in \reals^{m}$ and returns a set $S \subseteq A$ that maxmizes $f(S) - \sum_{j \in S} p_j$.
\end{itemize}

Both value and demand oracles are standard in the literature on combinatorial optimization. They also naturally come up in the context of combinatorial contracts. For instance, in the single-agent combinatorial actions case, a demand oracle corresponds to the agent's best response problem \cite{DuttingEFK21}. Indeed, there the agent's best response to a contract $\alpha$ is the set $S \subseteq A$ that maximizes $\alpha f(S) - c(S)$, or equivalently, dividing through $\alpha$ and interpreting $f(S)$ as the agent's valuation, the set $S$ that maximizes $f(S) - \sum_{j \in S} p_j$ at prices $p_j = c_j/\alpha$. More generally, in the multi-agent case, for a given contract $\contract$, a demand oracle corresponds to the best set of actions to take, had the principal to take the actions herself \cite{DuettingEFK23}. Specifically, Proposition~\ref{prop:potential-function} shows that, for any contract $\contract$, a demand oracle at prices $p_j/\alpha_i$ returns some equilibrium $S$.

\section{A Constant-Factor Approximation When No Agent Is Large}
\label{sec:no-agent-is-large}

As a first ingredient to our proof of Main Theorem~\ref{mainthm:1}, we show how to get a constant-factor approximation to the best contract in which no agent is large.

\begin{theorem}\label{thm:noagentislarge}
For every submodular success probability $f$, Algorithm~\ref{alg:cap} runs in polynomial time using polynomially many value queries, and returns a contract $\contract$ with a guarantee $\Lambda$ such that in any equilibrium $S$ of $\contract$  the principal's utility is at least $\Lambda$. 
The guarantee $\Lambda$ satisfies that: There exist constants $\phi, \kappa>0$ such that for any contract $\contract^\star$ and any equilibrium $S^\star$ of $\contract^\star$ such that $f(S^\star_i) \leq \phi\cdot f(S^\star)$ for all $i \in \agents$, it holds that  $\Lambda \geq \kappa \cdot (1-\sum_{i \in \agents} \alpha^\star_i)f(S^\star)$. 
\end{theorem}

\paragraph{Proof outline.}
The proof of Theorem~\ref{thm:noagentislarge} is constructed step-by-step in the subsequent subsections as detailed below.

In Section~\ref{sec:subset-stable} we define the notion of a {\em subset-stable} set with respect to a contract. Subset stability is a relaxation of an equilibrium, where the set is stable against removal of actions (any subset thereof) but not against the addition of actions. 
Subsequently, we establish a highly useful property of subset-stable sets, under submodular reward functions, as specified in the ``Doubling Lemma." This property suggests that, upon doubling the contract to all agents, the reward in {\em any} equilibrium of the doubled contract will be at least half of that in the subset-stable set.

However, identifying subset-stable sets with high value and relatively small payments remains a challenge due to the exponential number of incentive constraints.
To address this, in Section~\ref{sec:bundle-prices}, we identify a simple sufficient condition for a set to be subset stable with respect to a contract.
This condition specifies bundle pricing (which is not necessarily additive) for which a demand set satisfies this condition.

But now a new challenge emerges: answering this demand query with respect to the corresponding bundle prices is hard. To overcome this challenge, in Section~\ref{sec:approximate-demand} we devise an algorithm that gives an approximate demand set with respect to the bundle prices.

Finally, in Section~\ref{sec:combines-no-large}, we show how to transform the approximate demand set into a subset-stable set (with respect to a contract that pays at most a quarter of the reward) that approximates the optimal contract with no large agent.

\subsection{Subset Stability and the Doubling Lemma}\label{sec:subset-stable}

While it is typically insufficient to ensure stability against ``subset deviations" (namely, a deviation to a strict subset of the current set of actions), we show that ensuring such stability is sufficient in some precise approximate sense.
Before formalizing this claim, we introduce the notion of subset stability.

\begin{definition}[Subset stability]
\label{def:subset-deviations}
A set of actions $S$ is {\em subset stable} with respect to a contract $\contract$
if for all agents $i \in \agents$ and all $S'_i \subseteq S_i$ it holds that
\[
\alpha_i f(S_i,S_{-i}) - c(S_i) \geq \alpha_i f(S'_i,S_{-i}) - c(S'_i).
\]
\end{definition}

The following key lemma formalizes the argument that subset stability suffices approximately.  
This lemma will serve us more than once throughout the paper.

\begin{lemma}[The Doubling Lemma]
\label{lemma:anyequilibrium}
Suppose $f$ is submodular. Let $\epsilon > 0$ and let $\vec{\epsilon} = (\epsilon,\ldots,\epsilon) \in \reals_+^n$. Let $S$ be a subset-stable action set with respect to a contract $\contract$.
Then any equilibrium $\Seq{}$ with respect to $2 \contract + \vec{\epsilon}$ fulfills $f(\Seq{}) \geq \frac{1}{2} f(S)$.
\end{lemma}

We remark that Lemma~\ref{lemma:anyequilibrium} is quite strong. In particular, it holds for {\em any} subset-stable action set and {\em any} equilibrium with respect to $2 \contract + \vec{\epsilon}$. 
Indeed, this lemma holds with respect to submodular functions $f$, but fails for XOS functions, where there are instances in which any contract admits some bad equilibrium (see Proposition~\ref{prop:bad-xos}). We proceed with the proof of Lemma~\ref{lemma:anyequilibrium}.

\begin{proof}[Proof of Lemma~\ref{lemma:anyequilibrium}]
As $\Seq{}$ is an equilibrium with respect to $2 \contract + \vec{\epsilon}$, agent $i$ weakly prefers $\Seq{}_i$ over $\Seq{}_i \cup S_i$. That is,
\[
(2 \alpha_i + \epsilon) f(\Seq{}) - c(\Seq{}_i) \geq (2 \alpha_i + \epsilon) f(\Seq{}_i \cup S_i , \Seq{}_{-i}) - c(\Seq{}_i \cup S_i)
\]
and therefore
\[
(2 \alpha_i + \epsilon) \bigg( f(\Seq{}_i \cup S_i , \Seq{}_{-i}) - f(\Seq{}) \bigg) \leq c(\Seq{}_i \cup S_i) - c(\Seq{}_i) = c(S_i \setminus \Seq{}_i).
\]

Furthermore, since $S$ is subset stable with respect to $\contract$, agent $i$ weakly prefers $S_i$ over $\Seq{}_i \cap S_i \subseteq S_i$. So, we have
\[
\alpha_i f(S) - c(S_i) \geq \alpha_i f(\Seq{}_i \cap S_i , S_{-i}) - c(\Seq{}_i \cap S_i).
\]
This implies
\[
\alpha_i \bigg(f(S) - f(\Seq{}_i \cap S_i , S_{-i})\bigg) \geq c(S_i) - c(\Seq{}_i \cap S_i) = c(S_i \setminus \Seq{}_i). 
\]

In combination, we obtain
\begin{equation}
f(\Seq{}_i \cup S_i , \Seq{}_{-i}) - f(\Seq{}) \leq \frac{\alpha_i}{2\alpha_i + \epsilon} \left(f(S) - f(\Seq{}_i \cap S_i , S_{-i})\right) \leq \frac{1}{2} \left(f(S) - f(\Seq{}_i \cap S_i , S_{-i})\right),
\label{eq:dagger}
\end{equation}
where we use that the expression is non-negative by monotonicity of $f$.

Let $S^j_i = \Seq{}_i \cup S_i$ if $i \leq j$ and $S^j_i = \Seq{}_i$ otherwise. By this definition
\[
f(S) \leq f(S^n) = f(S^0) + \sum_{j=1}^n \left( f(S^j) - f(S^{j-1}) \right).
\]
Note that $S^0 = \Seq{}$. Furthermore, by submodularity of $f$, we have
\[
f(S^i) - f(S^{i-1}) = f(S_i \setminus \Seq{}_i \mid S^{i-1}) \leq  f(S_i \setminus \Seq{}_i \mid \Seq{})  = f(\Seq{}_i \cup S_i , \Seq{}_{-i}) - f(\Seq{}).
\]
Therefore, we have
\[
f(S) \leq f(\Seq{}) + \sum_{i=1}^n \left(f(\Seq{}_i \cup S_i , \Seq{}_{-i}) - f(\Seq{}) \right) \leq f(\Seq{}) + \sum_{i=1}^n \frac{1}{2} \left(f(S) - f(\Seq{}_i \cap S_i , S_{-i})\right) \leq f(\Seq{}) + \frac{1}{2} f(S \setminus \Seq{}),
\]
where in the second inequality we used Equation~\eqref{eq:dagger}. 
So, we have $f(S) \leq 2 f(\Seq{})$.
\end{proof}

\subsection{A Sufficient Condition for Subset Stability}
\label{sec:bundle-prices}

The following lemma defines a sufficient condition for a set to be subset stable for some contract. This is the key ingredient that drives our algorithm for finding subset stable sets with a relatively high value for which the corresponding contract does not pay more than half of the utility. 

\begin{lemma}
\label{lemma:combined}
Suppose $f$ is submodular. Fix $\gamma>0$ and let $S$ be a set such that for all $i$ and all $j \in S_i$, we have $f(j \mid S \setminus \{j\}) \geq \gamma \left(\sqrt{c(S_i)} - \sqrt{c(S_i \setminus \{j\}})\right)$. Then, $S$ is subset stable under the contract $\contract$ defined by
$\alpha_i = \frac{4}{\gamma} \sqrt{c(S_i)}$. \end{lemma}

The condition in Lemma~\ref{lemma:combined} has the following natural economic interpretation: For ``bundle prices'' $p(S) = \gamma \sum_{i \in \agents} \sqrt{c(S_i)}$, if
the (quasi-linear) ``utility'' $f(S)- p(S)$ cannot be improved upon by dropping any single action 
then the set is subset stable under contract $\contract$ (as defined in the lemma). 

One desirable property of this condition is that, unlike subset stability, it is easy to verify (we just need to check that we don't want to drop individual actions).
Our algorithms in Sections~\ref{sec:approximate-demand} and \ref{sec:combines-no-large} are inspired by this economic interpretation.

\begin{proof}[Proof of Lemma~\ref{lemma:combined}]

    Let $S_i' \subseteq S_i$, and let $X = S_i \setminus S'_i$. By submodularity, we have
\begin{equation}    
f(S_i \mid S_{-i}) - f(S_i' \mid S_{-i}) \geq \sum_{j \in X} f(j \mid S \setminus \{j\}). \label{eq:submod}
\end{equation}

Furthermore, each such term satisfies
\begin{equation}
f(j \mid S \setminus \{j\}) \geq \gamma( \sqrt{c(S_i)} - \sqrt{c(S_i \setminus \{j\})} ) = \gamma \frac{c(S_i) - c(S_i \setminus \{j\})}{ \sqrt{c(S_i)} + \sqrt{c(S_i \setminus \{j\})}} \geq \frac{\gamma c(\{j\})}{2 \sqrt{c(S_i)}}.\label{eq:marginal}
\end{equation}
By abbreviated multiplication we have that
\begin{equation}    
c(S_i) - c(S_i') = (\sqrt{c(S_i)} - \sqrt{c(S_i')})(\sqrt{c(S_i)} + \sqrt{c(S_i')}) \geq (\sqrt{c(S_i)} - \sqrt{c(S_i')})\sqrt{c(S_i)}. \label{eq:difc}
\end{equation}
Combining Equations~\eqref{eq:submod}, \eqref{eq:marginal}, and \eqref{eq:difc} we get:
\begin{eqnarray*}
f(S_i \mid S_{-i}) - f(S_i' \mid S_{-i})  & \stackrel{\eqref{eq:submod}}{\geq} & 
\sum_{j \in X} f(j \mid S \setminus \{j\})  \stackrel{\eqref{eq:marginal}}{\geq} \sum_{j \in X} \frac{\gamma c(\{j\})}{2 \sqrt{c(S_i)}} \\ &  =  & \frac{\gamma c(X)}{2 \sqrt{c(S_i)}} = \frac{\gamma}{2 \sqrt{c(S_i)}} (c(S_i) - c(S_i')) \stackrel{\eqref{eq:difc}}{\geq}  \frac{\gamma}{2} (\sqrt{c(S_i)} - \sqrt{c(S_i')}).
\end{eqnarray*}
This implies that
\[
\frac{c(S_i) - c(S_i')}{f(S_i \mid S_{-i}) - f(S_i' \mid S_{-i})} \leq \frac{c(S_i) - c(S_i')}{\frac{\gamma}{2} \left(\sqrt{c(S_i)} -  \sqrt{c(S_i')}\right)} = \frac{2}{\gamma} \left(\sqrt{c(S_i)} + \sqrt{c(S_i')} \right) \leq \frac{4}{\gamma} \sqrt{c(S_i)} = \alpha_i,
\]
where we use that $c(S_i') \leq c(S_i)$ because $S_i' \subseteq S_i$.

By rearrangement of the last inequality we get that $\alpha_i f(S) - c(S_i) \geq \alpha_i f(S_i', S_{-i}) - c(S_i')$, which concludes the proof. 
\end{proof}

\subsection{Approximate Demand Oracles for Bundle Prices}\label{sec:approximate-demand}
\newcommand{\q}{\ensuremath{q}}
\newcommand{\h}{\ensuremath{h}}
Lemma~\ref{lemma:combined} states that if $f(S) - p(S)$ cannot be increased by dropping actions, the set $S$ is subset stable under the contract $\contract$ defined there. 
Clearly, a set $S$ that maximizes $f(S) - p(S)$ fulfills this property. Motivated by this, in this section, we seek to approximately maximize $f(S) - p(S)$. This approximation algorithm will then be used as a subroutine in our final algorithm to compute a contract.

To keep notation clean and abstract away from our application, we consider the following problem. Let $\actions$ be any set partitioned into $\actions_1, \ldots, \actions_n$. Furthermore, let $\h\colon 2^\actions \to \mathbb{R}$ be any monotone submodular function and let $p\colon 2^\actions \to \mathbb{R}$ be defined by $p(S) = \sum_{i \in \agents} \sqrt{\sum_{j \in S_i} \q_j}$, where $\q_1, \ldots, \q_m \geq 0$ and $S_i = S \cap \actions_i$. Our goal is to find a set $S$ that maximizes $\h(S) - p(S)$.

\subsubsection{Single Agent}
\label{sec:single}
To simplify our optimization problem, we first consider the case that $n = 1$. That is, our goal is now to maximize $\h(S) - \sqrt{\sum_{j \in S} \q_j}$. For this problem, we propose Algorithm~\ref{alg:single}. The idea behind this algorithm is as follows. Suppose we knew the value $\h(S^\star)$ for the optimal solution $S^\star$. Then define additive prices $\tilde{p}_j \approx \frac{\q_j}{\h(S^\star)}$ and compute an approximate demand set with respect to these additive prices. 

\begin{algorithm}
\caption{Approximate demand oracle for 
bundle prices for a single agent}\label{alg:single}
   \hspace*{\algorithmicindent} \textbf{Input:}  Costs $\q_1,\ldots,\q_m \in \reals_{\geq 0}$, and value oracle access to a submodular function $\h:2^A \rightarrow \reals_{\geq 0}$ \\
    \hspace*{\algorithmicindent} \textbf{Output:}  A set $S$ 
\begin{algorithmic}[1]
\For{$j \in M$ and $k\in\{1,\ldots ,\lceil\log m\rceil \}$} 
\State Let $x = 2^k \cdot \h(j)$ 
\State $S^{x} = \emptyset$
\State Set $\tilde{p}_{j'} = \frac{18 \q_{j'}}{x}$ for all $j' \in M$
\While{there is $j' \in M \setminus S^{x}$ such that $\h(j' \mid S^{x}) - \tilde{p}_{j'} \geq 0$} 
\State Add to $S^{x}$ the element $j'$ that maximizes $\h(j' \mid S^{x}) - \tilde{p}_{j'}$
\EndWhile
\EndFor
\State \Return $S^{x}$ for which $\h(S^{x}) - \sqrt{\q( S^{x})}$ is maximized
\end{algorithmic}
\end{algorithm}

\begin{theorem}
For $\rho =\frac{1}{6}$,
Algorithm~\ref{alg:single} computes a set $S$ satisfying
\[
\h(S) - \sqrt{\sum_{j \in S} \q_j} \geq \max_{S^\star}\left( \rho \cdot \h(S^\star) - \sqrt{\sum_{j \in S^\star} \q_j}\right).
\]
\label{thm:alg-single}
\end{theorem}

\begin{proof}
Let $S^\star$ maximize $\rho \cdot \h(S^\star) - \sqrt{\sum_{j \in S^\star} \q_j}$. If $S^\star =\emptyset$ then the guarantee is trivial, since the algorithm returns a set with a non-negative value of $\h(S) - \sqrt{\sum_{j \in S} \q_j}$.
Otherwise, let $j^\star = \arg\max_{j \in S^\star} \h(j)$ note that
$ \h(j^\star) \leq \h(S^\star) \leq m  \cdot \h(j^\star)$, where the second inequality follows by subadditivity.
So, there is an execution of the algorithm in which $x/2 \leq \h(S^\star) \leq x$ since when the algorithm considers $j=j^\star$, then for $k=1$ it holds that $x/2 \leq \h(j^\star)$ while for $k=\lceil\log m \rceil$ it holds that $x \geq m \cdot \h(j^\star)$.

\begin{observation}\label{obs:fp}
The algorithm computes a set $S$ such that $\h(S) \geq \sum_{j \in S} \tilde{p}_j$.
\end{observation}

\begin{lemma}\label{lem:fss}
The algorithm computes a set $S$ such that 
$\h(S) \geq  \frac{1}{2} \h(S^\star)$.
\end{lemma}

\begin{proof}
Upon termination it holds that
$\h(j \mid S) - \tilde{p}_j \leq 0$ for all $j \in M$.

So
\begin{eqnarray}
0  & \geq &  \max_{j \in M} (\h(j \mid S) - \tilde{p}_j) \nonumber   \geq  \frac{1}{\lvert S^\star \rvert} \sum_{j \in S^\star} (\h(j \mid S) - \tilde{p}_j)  \nonumber \\ & \geq & \frac{1}{\lvert S^\star \rvert} \left (\h(S^\star \cup S) - \h(S) - \sum_{j \in S^\star} \tilde{p}_j\right)  \geq  \frac{1}{\lvert S^\star \rvert} \left (\h(S^\star) - \h(S) - \sum_{j \in S^\star} \tilde{p}_j\right),  \nonumber 
\end{eqnarray}
where the third inequality is by submodularity of $\h$, and the fourth inequality is by monotonicity of $\h$.
Combined with the fact that $|S^\star| \geq 1$, this implies that $\h(S) \geq  \h(S^\star) - \sum_{j \in S^\star} \tilde{p}_j$. Furthermore
\[
\sum_{j \in S^\star} \tilde{p}_j = \sum_{j \in S^\star} \frac{18 \q_j}{x} \leq \frac{18 (\rho \cdot  \h(S^\star))^2}{x} \leq \frac{1}{2} \h(S^\star),
\]
where the first inequality is since $\rho \cdot \h(S^\star) - \sqrt{\sum_{j \in S^\star} \q_j} \geq 0$, and the second inequality is since $x \geq \h(S^\star)$ and by the value pf $\rho$.
\end{proof}

By combining  (1) Observation~\ref{obs:fp},  (2) that $x\leq 2\h(S^\star)$, and (3) Lemma~\ref{lem:fss} we get that:
\[
\sum_{j \in S} \q_j = \frac{x}{18} \sum_{j \in S} \tilde{p}_j \leq \frac{2 \h(S^\star)}{18} \h(S) \leq \frac{4}{18} \h(S)^2,
\]
which by rearrangement gives us that $\sqrt{\sum_{j \in S} \q_j} \leq \frac{\sqrt{2}}{3} \h(S)$.
Therefore, it holds that
\[
\h(S) - \sqrt{\sum_{j \in S} \q_j} \geq \frac{3 - \sqrt{2}}{3} \h(S) \geq \rho \cdot \h(S^\star) \geq \rho \cdot \h(S^\star) - \sqrt{\sum_{j \in S^\star} \q_j},
\]
which concludes the proof. 
\end{proof}

\subsubsection{Multiple Agents}
\label{sec:multi}

We now turn to the actual problem of maximizing $\h(S) - \sum_{i \in \agents} \sqrt{\sum_{j \in S_i} \q_j}$. Algorithm~\ref{alg:multi} repeatedly calls Algorithm~\ref{alg:single}. In these calls, it sets $\h'$ to be the marginal increases and modifies $\q_1, \ldots, \q_m$ by setting $\tilde{\q}_j = 4 \q_j$.

\begin{algorithm}
\caption{Approximate demand oracle for 
bundle prices for multiple agents}\label{alg:multi}
   \hspace*{\algorithmicindent} \textbf{Input:}  Costs $\q_1,\ldots,\q_m \in \reals_{\geq 0}$, and value oracle access to a submodular function $\h:2^A \rightarrow \reals_{\geq 0}$ \\
    \hspace*{\algorithmicindent} \textbf{Output:}  A set $S$ 
\begin{algorithmic}[1]
\State Let $\tilde{\q}_j = 4 \q_j$ for every $j \in M$
\For {$t = 1, \ldots, n$} 
\For{ $i = 1, \ldots, n$}
\State Let $\h':A_i \rightarrow \reals_{\geq 0}$ be such that $\h'(S) = \h(S^{(t, i)} \mid S^{(1)} \cup \ldots \cup S^{(t-1)})$
\newline  \Comment{Note that a value query to $\h'$ can be computed using two value queries to $\h$}
\newline \Comment{If $\h$ is monotone and submodular then   $\h'$ is monotone, submodular, and normalized }
\State Call Algorithm~\ref{alg:single} on $\tilde{\q}_1,\ldots,\tilde{\q}_m$ and $\h'$ to obtain $S^{(t, i)} \subseteq A_i$
\EndFor
\State Let $S^{(t)} $ be the set that maximizes $  
\h(S^{(t, i)} \mid S^{(1)} \cup \ldots \cup S^{(t-1)}) - \sqrt{\sum_{j \in S^{(t,i)}} \tilde{\q}_j}
$ over all $i$
\EndFor
\State \Return  $S = S^{(1)} \cup \ldots \cup S^{(n)}$.
\end{algorithmic}
\end{algorithm}

\begin{theorem}\label{thm:alg-multi}
For the parameter $\rho$ of Theorem~\ref{thm:alg-single},  
the set $S$ returned by  Algorithm~\ref{alg:multi} fulfills
\[
\h(S) - \sum_{i \in \agents} \sqrt{\sum_{j \in S_i} \q_j} \geq \max_{S^\star} \left( \frac{\rho}{2} \h(S^\star) - \sum_{i \in \agents} \sqrt{\sum_{j \in S^\star_i} \q_j} \right).
\]
\end{theorem}

\begin{proof}
Let $(S_1^\star, \ldots, S_n^\star)$ be any sets with $S_i^\star \subseteq A_i$.

Write $\tilde{p}(X) = \sum_{i \in \agents} \sqrt{\sum_{j \in X_i} \tilde{\q}_j}$. Note that by this definition $\tilde{p}(S^\star) = \sum_{i \in \agents}\tilde{p}(S_i^\star)$ because $(S_1^\star, \ldots, S_n^\star)$ are disjoint and each belongs to a different $A_i$. Furthermore $\tilde{p}(S) \leq \sum_{t=1}^{n} \tilde{p}(S^{(t)})$ because $\tilde{p}$ is subadditive.

For every $t = 1,\ldots,n$, by Theorem~\ref{thm:alg-single}
 we have
\begin{align}
\h(S^{(t)} \mid S^{(1)} \cup \ldots \cup S^{(t-1)}) - \tilde{p}(S^{(t)}) & \geq \max_{i \in \agents} \left( \rho \h(S_i^\star \mid S^{(1)} \cup \ldots \cup S^{(t-1)}) - \tilde{p}(S_i^\star)  \right) \nonumber \\
& \geq \frac{1}{n} \sum_{i \in \agents} \left(\rho \h(S_i^\star \mid S^{(1)} \cup \ldots \cup S^{(t-1)}) - \tilde{p}(S_i^\star) \right) \nonumber\\
& \geq \frac{1}{n} \sum_{i \in \agents} \left( \rho \h(S_i^\star \mid S^{(1)} \cup \ldots \cup S^{(t-1)} \cup S_1^\star \cup \ldots \cup S_{i-1}^\star) - \tilde{p}(S_i^\star) \right) \nonumber\\
& = \frac{1}{n} \left(\rho \h(S_1^\star \cup \ldots \cup S_n^\star \mid S^{(1)} \cup \ldots \cup S^{(t-1)}) - \tilde{p}(S^\star) \right) \nonumber\\
& \geq \frac{1}{n} \left( \rho \left( \h(S_1^\star \cup \ldots \cup S_n^\star) - \h(S^{(1)} \cup \ldots \cup S^{(t-1)}) \right) - \tilde{p}(S^\star) \right)\nonumber\\
& \geq \frac{1}{n} \left(\rho \left( \h(S^\star) - \h(S) \right) - \tilde{p}(S^\star) \right) \label{eq:pfs}
\end{align}
Thus,
\begin{align*}
\h(S) - \tilde{p}(S) & \geq \h(S) - \sum_{t = 1}^{n} \tilde{p}(S^{(t)})\\
&= \sum_{t=1}^{n} \left( \h(S^{(t)} \mid S^{(1)} \cup \ldots \cup S^{(t-1)}) - \tilde{p}(S^{(t)}) \right) \\
& \stackrel{\eqref{eq:pfs}}{\geq} \sum_{t=1}^{n} \frac{1}{n} \left(\rho \left( \h(S^\star) - \h(S) \right) - \tilde{p}(S^\star) \right)\\
& = \rho \h(S^\star) - \rho \h(S) - \tilde{p}(S^\star)
\end{align*}
Therefore. since $\rho\leq 1$, we get that $$2 \h(S) - \tilde{p}(S) \geq 
(1+\rho) \h(S) - \tilde{p}(S)
\geq \rho \h(S^\star) - \tilde{p}(S^\star)$$ and thus $\h(S) - \frac{1}{2} \tilde{p}(S) \geq \frac{\rho}{2} \h(S^\star) - \frac{1}{2} \tilde{p}(S^\star)$,
which by plugging the definitions of $\tilde{p},\tilde{q}$ concludes the proof.
\end{proof}

\subsection{Approximation Algorithm Under No Large Agent}\label{sec:combines-no-large}

We are now ready to present our algorithm (Algorithm~\ref{alg:cap}) that finds an approximate optimal contract among all contracts under the assumption of no large agents.

\begin{algorithm}
\caption{Approximate optimal contract under no large agent}\label{alg:cap}
   \hspace*{\algorithmicindent} \textbf{Parameters:}  $\rho \in [0,1]$, $\epsilon \in (0,\frac{1}{4n})$ \\
   \hspace*{\algorithmicindent} \textbf{Input:}  Costs $c_1,\ldots,c_m \in \reals_{\geq 0}$, and value oracle access to a submodular function $f:2^A \rightarrow \reals_{\geq 0}$  \\
    \hspace*{\algorithmicindent} \textbf{Output:} A contract $\contract$ with a guarantee of principal's utility of $\Lambda$ for any equilibrium of $\contract$
\begin{algorithmic}[1]
\For{$j^\star \in \actions$ with $f(j^\star)>0$ and $t = 0, \ldots, \lceil \log m \rceil$}

\State Let $\xi = \frac{\rho^2}{512} \cdot f(j^\star)\cdot 2^t$ and  $\gamma = \sqrt{32\xi}$ \label{step:gamma}
\State Let $\hat{f}:2^A\rightarrow \reals_{\geq 0}$ be such that $\hat{f}(S) = \min\{f(S), \xi\}$
\State Let $p:2^A\rightarrow \reals_{\geq 0}$ be defined by $p(S) = \gamma \sum_{i \in \agents} \sqrt{\sum_{j \in S_i} c_j}$
\State  Apply Algorithm~\ref{alg:multi} on $\h = \hat{f}$ and $\q_j = \gamma^2 c_j$ to maximize $\hat{f}(S) - p(S)$, let $S^\gamma$ be its output  \label{step:subroutine}
\While{ exist $j\in S$ such that  $\hat{f}(j \mid S^\gamma \setminus \{ j \}) < p(S^\gamma) - p(S^\gamma \setminus \{ j \})$ \label{step:remove}}  
\State $S^\gamma =S^\gamma\setminus \{j\}$ 
\EndWhile\label{step:while2}
\State Let $\alpha^\gamma_i = \frac{4}{\gamma} \sqrt{c(\hat{S}_i)}$
\EndFor
\State Let $\gamma^\star = \arg\max_\gamma f(S^\gamma)$ \label{step:gamma2}
\State\Return{$\contract = 2\contract^{\gamma^\star} +\vec{\epsilon} $,  $ \Lambda = \frac{f(S^{\gamma^\star})}{8}$}

\end{algorithmic}
\end{algorithm}

Recall that in this section $\contract^\star,S^\star$ is the best contract and the corresponding best equilibrium for which $f(S^\star_i) \leq \phi f(S^\star)$ for all $i\in \agents$. We assume that $f(S^\star)>0$ (as otherwise the instance is redundant and our algorithm always returns a contract with non-negative utility). Furthermore, $\sum_{i \in \agents} \alpha^\star_i \leq 1$ because otherwise the principal utility would be negative.

Algorithm~\ref{alg:cap}
first guesses an estimate of $f(S^\star)$, which is captured by $\xi$. 
Given this estimate, the algorithm defines a new instance with a reward function that is capped at $\xi$, 
and applies
Algorithm~\ref{alg:multi} for this modified instance which returns a set $S$. From this returned set, the algorithm removes actions with negative marginal contribution with respect to the bundle prices defined by $p(\cdot)$.  Among all such created sets, the algorithm returns the one with the maximum value of $f(\cdot)$ (over all choices of $\xi$). 
 
 We analyze the guarantees of Algorithm~\ref{alg:cap} with parameter $\rho$ set as in Theorem~\ref{thm:alg-single}. 
We first claim that the algorithm is well defined. For this we note that since $f(S^\star) > 0$, there exists $j^\star\in\actions$ with $f(j^\star)>0$, which means that $\gamma^\star$ is well defined (as it is an argmax of a non-empty set). We also note that a value query to $\hat{f}$ can be computed using a value query to $f$ in order to call Algorithm~\ref{alg:multi} in Step~\ref{step:subroutine}, and that the function $\hat{f}$ is monotone, submodular,  and normalized since $f$ is monotone, submodular, and normalized, which satisfies 
the requirements of Algorithm~\ref{alg:multi}.
 
We now observe that there exists an iteration for which $\xi$ is close to $ \frac{\rho^2}{512} f(S^\star)$ 
\begin{observation}  \label{obs:xi}   
There exist $j^\star\in \actions$, and $t\in \left\{ 0, \ldots,  \lceil \log m \rceil \right\}$  such that $ \frac{\rho^2}{1024} f(S^\star)\leq \xi \leq \frac{\rho^2}{512} f(S^\star)$.
\end{observation}
\begin{proof}
    Since $f(S^\star) > 0 $ then $S^\star \neq \emptyset$ and for the iteration with $j^\star \in \arg\max_{j\in S^\star} f(j)$ it holds for $t=0$ then $\xi = \frac{\rho^2}{512} f(j^\star)  \leq \frac{\rho^2}{512} f(S^\star) $ and for $t=\lceil \log m \rceil $ it holds that  $\xi =\frac{\rho^2}{512}  f(j^\star) \cdot 2^{\lceil \log m \rceil}\geq \frac{\rho^2}{512}  f(j^\star) \cdot  m  \geq \frac{\rho^2}{512}  f(S^\star)    $. Thus there must be a value of $t$ is each time multiplied by a factor of $2$, $\xi$ must be in the range for one iteration. 
\end{proof}
  Let $\gamma^\star$ be its corresponding value calculated in Step~\ref{step:gamma2}, let $\hat{S} = S^{\gamma^\star}$, let $\xi^\star = \frac{\gamma^{\star2}}{32}$ be its corresponding value of $\xi$,  and let $\contract = (\alpha_1,\ldots,\alpha_n) = (\alpha^{\gamma^\star}_1,\ldots,\alpha^{\gamma^\star}_n)$ be their corresponding values at the termination of Algorithm~\ref{alg:cap}. 
We first bound the sum of $\alpha$s.
\begin{lemma} It holds that
   $\sum_{i \in \agents} \alpha_i \leq \frac{1}{4}$. \label{lem:sum-alpha}
\end{lemma}
\begin{proof}
 
For every $i\in\agents$ it holds that
\begin{eqnarray}
\hat{f}(\hat{S}_i \mid \hat{S}_{-i}) &\geq & \sum_{j \in \hat{S}_i} \hat{f}(j \mid \hat{S} \setminus \{j\}) \geq \sum_{j \in \hat{S}_i} \gamma^{\star} (\sqrt{c(\hat{S}_i)} - \sqrt{c(\hat{S}_i \setminus \{j\})}) \nonumber \\ &= & \gamma^{\star} \sum_{j \in \hat{S}_i} \frac{c(\hat{S}_i) - c(\hat{S}_i \setminus \{j\})}{\sqrt{c(\hat{S}_i)} + \sqrt{c(\hat{S}_i \setminus \{j\})}} 
\geq \gamma^{\star}\sum_{j \in \hat{S}_i}  \frac{c(j)}{2 \sqrt{c(\hat{S}_i)}} = \frac{\gamma^{\star}}{2} \sqrt{c(\hat{S}_i)}, \label{eq:fs}
\end{eqnarray}
where the first inequality is by submodularity, the second inequality by Step~\ref{step:remove} of the algorithm, and the third inequality is by additivity and monotonicity of the function $c$.

Thus, 
$$\sum_i \alpha_i  = \sum_i  \frac{4}{\gamma^{\star}} \sqrt{c(\hat{S}_i)} \leq  \sum_i  \frac{4}{\gamma^{\star}} \cdot \frac{2}{\gamma^{\star}} \hat{f}(\hat{S}_i \mid \hat{S}_{-i}) \leq \frac{8}{\gamma^{\star2}}\hat{f}(\hat{S}) \leq \frac{8}{\gamma^{\star2}}\xi^\star  =   \frac{1}{4},$$
where the first inequality is by Inequality~\eqref{eq:fs}, the second inequality is by submodularity, and the third inequality is by the definition of $\hat{f}$.
\end{proof}  
We next show that $\hat{S}$ is subset stable.
\begin{lemma} \label{lem:s-stable}
The set    $\hat{S}$ is subset stable with respect to contract $\contract$.
\end{lemma}
\begin{proof}

 For all agents $i\in \agents$ and all actions $j \in A_i$ it holds that
\begin{equation}    
f(j \mid \hat{S} \setminus \{j\}) \geq \hat{f}(j \mid \hat{S} \setminus \{j\}) \geq p(\hat{S}_i) - p(\hat{S}_i \setminus \{j\}) =\gamma^\star \cdot \left(\sqrt{c(\hat{S}_i)} - \sqrt{c(\hat{S}_i \setminus \{j\})}\right) ,
\label{eq:alpha-stable}\end{equation}

where the first inequality is since capping a function by a constant can only decrease the marginal utility, and the second inequality is ensured by Step~\ref{step:remove}.
Thus, the lemma follows by Lemma~\ref{lemma:combined}, Inequality~\eqref{eq:alpha-stable} and the definition of $\alpha_i$.
\end{proof}
We next bound the value of $f(\hat{S})$.
\begin{lemma} If $S^\star$ is an equilibrium for $\contract^\star$ fulfilling $\sum_{i \in \agents} \alpha_i^\star \leq 1$ and $f(S^\star_i) \leq \phi f(S^\star)$ for all $i\in \agents$, then  
   $f(\hat{S}) \geq \left(\frac{\rho^3 }{6144} -\frac{\rho \cdot \phi}{2}\right) f(S^\star)$. \label{lem:s-value}
\end{lemma}
\begin{proof}
If $\phi \geq \frac{\rho^2}{3072}$ or $f(S^\star)=0$ then the claim holds trivially since $\left(\frac{\rho^3 }{6144} -\frac{\rho \cdot \phi}{2}\right) f(S^\star) \leq 0$.
Since $S^\star$ is an equilibrium, no agent $i$ can increase his utility by dropping all their actions. That is $\alpha_i^\star f(S^\star) - c(S^\star_i) \geq \alpha_i^\star f(S^\star \setminus S^\star_i)$. So, it holds that $\alpha_i^\star \geq \frac{c(S^\star_i)}{f(S^\star_i \mid S^\star \setminus S^\star_i)}$. Furthermore, $\sum_{i \in \agents} \alpha^\star_i \leq 1$.
Using the Cauchy--Schwarz inequality $\left(\sum_i x_i y_i\right)^2 \leq \sum_i x_i^2 \sum_i y_i^2$ with $x_i = \sqrt{\frac{c(S^\star_i)}{f(S^\star_i \mid S^\star \setminus S^\star_i)}}$ and $y_i = \sqrt{f(S^\star_i \mid S^\star \setminus S^\star_i)}$, we get
\[
\left(\sum_i \sqrt{c(S^\star_i)}\right)^2 \leq \sum_i \frac{c(S^\star_i)}{f(S^\star_i \mid S^\star \setminus S^\star_i)} \sum_i f(S^\star_i \mid S^\star \setminus S^\star_i) \leq \sum_i \alpha_i^\star f(S^\star) \leq f(S^\star),
\]
where the second inequality is by submodularity of $f$.
Therefore $\sum_{i \in \agents} \sqrt{c( S^\star_i)} \leq \sqrt{f(S^{\star})}$.

By Observation~\ref{obs:xi} we know that there is iterations of $j^\star,t$ of Algorithm~\ref{alg:cap} for which their corresponding $\xi \in [\frac{\rho^2}{1024} f(S^\star) , \frac{\rho^2}{512} f(S^\star)]$.
We denote this value of $\xi$ by $\xi^{\star\star}$, and let $\gamma^{\star\star}$ be its corresponding value of $\gamma$ (i.e., $\gamma^{\star\star}= \sqrt{32\xi^{\star\star}}$). 
For $\gamma^{\star\star}$, let $S^1,S^2$ be the values of $S^{\gamma^{\star\star}}$ before and after Steps~\ref{step:remove}-\ref{step:while2} respectively.

Let $B_1, \ldots, B_\ell \subseteq \agents$ be an arbitrary partition of the set of agents $\agents$ such that $\xi^{\star\star} - \max_{i \in \agents} f(S^\star_i) \leq f(\bigcup_{i \in B_k} S^\star_i) \leq \xi^{\star\star}$ for all $k<\ell$, and  $f(\bigcup_{i \in B_k} S^\star_i) \leq \xi^{\star\star}$ for $k=\ell$. 
Such a partition exists and can be created by bundling $S_i^\star$ greedily until each bundle has a value in the desired range. This greedy process is well defined since $\xi^{\star\star} > \max_i f(S^\star_i)$ which holds since $\phi < \frac{\rho^2}{3072}$.

Note that $\ell \geq \frac{f(S^\star)}{\xi^{\star\star}}$ because otherwise by subadditivity $f(S^\star) = f(\bigcup_{k=1}^\ell \bigcup_{i \in B_k} S^\star_i) \leq \sum_{k=1}^\ell f(\bigcup_{i \in B_k} S^\star_i) \leq \ell \xi^{\star\star} < f(S^\star)$.

As $B_1, \ldots, B_\ell$ partition $\agents$, it holds that $\sum_{k<\ell} \sum_{i \in B_k} \sqrt{c( S^\star_i)}  \leq \sum_{i=1}^n \sqrt{c( S^\star_i)} \leq  \sqrt{f(S^{\star})}$, and there must exist $k<\ell$ such that $\sum_{i \in B_k} \sqrt{c( S^\star_i)} \leq \frac{1}{\ell-1} \sqrt{f(S^{\star})}$. Let $S^{\star\star} = \bigcup_{i \in B_k} S^\star_i$. 
Note that 
\begin{equation}
    \hat{f}(S^{\star\star}) = f(S^{\star\star}) \geq \xi^{\star\star} - \max_{i \in \agents} f(S^\star_i) ,\label{eq:sstar-lower}
\end{equation} 
and 
\begin{equation}
    \sum_{i \in \agents} p(S^{\star\star}_i) \leq \frac{\gamma}{\ell-1} \sqrt{f(S^\star)} \leq \frac{\rho \xi^{\star\star}}{3}. \label{eq:ps-upper}
\end{equation}
Therefore, it holds that 
\begin{eqnarray}
    f(\hat{S}) & \geq & f(S^2) \geq  \hat{f}({S^2}) \geq \hat{f}(S^2) - \sum_{i\in \agents} p(S^2_i) \geq \hat{f}(S^1) - \sum_{i \in \agents} p(S^1_i)  \nonumber \\ &= & \hat{f}(S^1) - \sum_{i \in \agents} \gamma \cdot \sqrt{c(S^1_i)} \geq 
    \frac{\rho}{2} \cdot \hat{f}(S^{\star\star}) - \sum_{i \in \agents} \gamma \cdot \sqrt{c(S^{\star\star}_i)} \nonumber \\ & =& \frac{\rho}{2} \cdot \hat{f}(S^{\star\star}) - \sum_{i \in \agents}  p(S^{\star\star}_i)   \geq \frac{ \rho \xi^{\star\star}}{6} - \frac{\rho}{2} \max_{i \in \agents} f(S_i^\star) \geq \left(\frac{\rho^3 }{6144} -\frac{\rho \cdot \phi}{2}\right) f(S^\star) , \nonumber
\end{eqnarray}
where the fourth inequality is by Theorem~\ref{thm:alg-multi}, and the last inequality is by Equations~\eqref{eq:sstar-lower} and \eqref{eq:ps-upper}.
\end{proof}
We are now ready to prove the main theorem of this section.
\begin{proof}[Proof of Theorem~\ref{thm:noagentislarge}]
The algorithm clearly runs in polynomial time using polynomially many value queries. By Lemma~\ref{lem:sum-alpha} we know that the sum of alphas at the termination is at most $\frac{1}{4}$, which implies that under contract $2\contract+\vec{\epsilon}$ the principal pays at most $2 \cdot \frac{1}{4}+ n \cdot \epsilon \leq \frac{3}{4}$.
Let $S$ be some equilibrium with respect to contract $2\contract+\vec{\epsilon}$ then  by Lemma~\ref{lem:s-stable} and Lemma~\ref{lemma:anyequilibrium}, we get that $f(S) \geq \frac{1}{2} \cdot f(\hat{S}) $,
which means that the principal's utility is $$  u_p(S,2\contract+\vec{\epsilon})  \geq (1-\sum_i \alpha_i) f(S) \geq \frac{1}{8} f(\hat{S}) = \Lambda.$$
By Lemma~\ref{lem:s-value}, it holds that $\Lambda = \frac{1}{8}f(\hat{S}) \geq \frac{1}{8}\left(\frac{\rho^3 }{6144} -\frac{\rho \cdot \phi}{2}\right) f(S^\star) $, which concludes the proof. 
\end{proof}

\section{Approximating the Single Agent Benchmark}\label{sec:single-agent}
In this section, as our second ingredient for Main Theorem~\ref{mainthm:1}, we show how to obtain a constant factor approximation against the best single-agent contract.

\begin{theorem}
\label{thm:robust-single-agent}
For every submodular success probability $f$, Algorithm~\ref{alg:robust-single-agent} runs in polynomial time using polynomially many value and demand queries, and returns a contract $\contract$ with a guarantee $\Lambda$ such that in any equilibrium $S$ of $\contract$  the principal's utility is at least $\Lambda$. 
The guarantee $\Lambda$ satisfies that: For any contract $\contract^\star$ and any equilibrium $S^\star$ of $\contract^\star$ such that $\alpha^\star_i = 0$ and $S^\star_i = \emptyset$ for all but one $i \in \agents$, it holds that  $\Lambda \geq \frac{1}{24} \cdot (1-\sum_{i \in \agents} \alpha^\star_i)f(S^\star)$.
\end{theorem}

In Section~\ref{sec:fptas-single-agent} we give a FPTAS for the single-agent problem for 
general $f$, with value and demand oracle access. This result is 
of independent interest.
Previously, only a weakly poly-time algorithm was known \cite{DuttingEFK21}.

Then, in Section~\ref{sec:making-single-agent-robust}, we extend this result to a robust approximation 
in the multi-agent case with submodular $f$; namely, a contract for which {\em every} induced equilibrium obtains the desired approximation. The challenge
arises from the existence of equilibria in which other agents take actions, even when their $\alpha_i$ is zero, which may in turn affect other agents' actions.
Our argument relies on the Doubling Lemma (Lemma~\ref{lemma:anyequilibrium}).

\subsection{FPTAS for a Single Agent}\label{sec:fptas-single-agent}

Consider a single-agent instance defined by a set $\actions$ of $m$ actions, a success probability function $f:2^\actions \rightarrow [0,1]$, and additive costs $c$. Let $\alpha^\star$, $S^\star$ be the optimal contract and the optimal set of actions.

Below we present an FPTAS for the optimal contract problem with general rewards $f$, with value and demand oracle access.
Previously, only a weakly-polynomial time FPTAS 
was known  \cite{DuttingEFK21}. Our result improves on this 
by showing that an arbitrarily good approximation to the optimal contract can be found in poly-time, independent of the representation size of the input. 

Our result is best possible in two ways: First, it is known that  computing an optimal contract is \textsf{NP}-hard for submodular $f$ \cite{DuttingEFK21}.
Second, it was recently shown that computing an optimal contract for submodular $f$ requires exponentially many demand oracle calls \cite{DuettingFGR24}.

\begin{theorem}[FPTAS for single-agent problem]\label{thm:FPTAS}
Consider the single-agent problem with 
general $f$. Then Algorithm~\ref{alg:fptas} gives a $(1-\epsilon)$-approximation to the optimal principal utility with $O\left(\frac{m^2}{\epsilon} \right)$ many value and demand queries.
\end{theorem}

The main challenge of the single agent case is that the optimal $\alpha^\star$ could be very close to $1$. Indeed, if we knew that $\alpha^\star$ was bounded away from $1$, we could find a close to optimal contract through a fine-enough discretization of the contract space.
To bound $\alpha^\star$, we use that, by a reduction to the non-combinatorial model of \cite{DuttingRT19}, the gap between the maximum welfare and the principal's utility under the optimal contract is at most $2^m$.

\begin{observation}[\cite{DuttingRT19}]
    For every reward function $f:2^\actions \rightarrow \reals_{\geq 0}$ over $|A|=m$ actions and a single agent with costs $c_1,\ldots,c_m$, there exists a contract $\alpha$ that guarantees a utility for the principal of at least  $\frac{\opt}{2^m}$, where $\opt=\max_{S} (f(S)-c(S))$. \label{obs:welfare-single}
\end{observation}

Using this observation, we now show our key lemma, which sets lower and upper bounds on $\alpha^\star$.
\begin{lemma}\label{lem:4m}
    Consider the single-agent problem with reward function $f$.  Let $\alpha^\star $, $S^\star $ be the optimal contract and the optimal set of actions, respectively, let $j^\star \in \arg\max_{j \in S^\star} c_j$, and let $\opt=\max_{S} (f(S)-c(S)) > 0$.
    Then we have
    $\alpha_{\min} \leq \alpha^\star \leq  \alpha_{\max}$, where $\alpha_{\min} =1- \frac{\opt}{ c_{j^\star} + \opt}$ and $\alpha_{\max} = 1- \frac{\opt}{m\cdot 2^m (c_{j^\star} + \opt )}$.
\end{lemma}

\begin{proof}
The set that maximizes $f(S)-c(S)$ is also the best response of the agent under contract $\alpha=1$.
We note that since $\opt>0$, it must be that $\alpha^\star <1 $.
By Observation~\ref{obs:welfare-single}, it holds that 
\begin{equation}
    \frac{\opt}{2^m}\leq (1-\alpha^\star)  f(S^\star)   \leq   f(S^\star) -c(S^\star) \leq \opt,\label{eq:welfare3}
\end{equation}
where the second inequality is since the utility of the principal under contract $\alpha^\star$ is bounded by the welfare generated by $S^\star$ as the utility of the agent is non-negative.

We first show that $\alpha^\star \geq \alpha_{\min}$. It holds that 
$$ \alpha^\star \cdot (  c(S^\star) + \opt  )-c(S^\star) \geq \alpha^\star \cdot (  c(S^\star) + f(S^\star)-c(S^\star) )-c(S^\star) =  \alpha^\star \cdot f(S^\star) -c(S^\star) \geq 0, $$
where the first inequality is by Inequality~\eqref{eq:welfare3}, and the last inequality is since the utility of the agent under $\alpha^\star$ and $S^\star$ is non-negative.

By rearranging, we get that $$ \alpha^\star \geq \frac{c(S^\star)}{ c(S^\star) + \opt} \geq \frac{c_{j^\star}}{c_{j^\star} +\opt } = 1- \frac{\opt}{ c_{j^\star} + \opt},
$$
as needed.

To show that $\alpha^\star \leq \alpha_{\max}$, observe that 
$$ (1-\alpha^\star) (c(S^\star) + \opt) \stackrel{\eqref{eq:welfare3}}{\geq }  (1-\alpha^\star) (c(S^\star) + f(S^\star) - c(S^\star)) =  (1-\alpha^\star) f(S^\star) \stackrel{\eqref{eq:welfare3}}{\geq } \frac{\opt}{2^m}.  $$
Rearranging this, we obtain
$$  \alpha^\star \leq 1- \frac{\opt}{2^m (c(S^\star) + \opt )} \leq 1- \frac{\opt}{m\cdot 2^m (c_{j^\star} + \opt )}  ,$$
 which concludes the proof.
\end{proof}

\begin{algorithm}
\caption{FPTAS for a single agent using value and demand oracles}\label{alg:fptas}
   \hspace*{\algorithmicindent} \textbf{Parameter:}  $\epsilon \in (0,1) $ \\
   \hspace*{\algorithmicindent} \textbf{Input:}  Costs $c_1,\ldots,c_m \in \reals_{\geq 0}$, value and demand oracle access to a function $f:2^A \rightarrow \reals_{\geq 0}$  \\
    \hspace*{\algorithmicindent} \textbf{Output:}  A contract $\alpha$ and a best response set $S$ 
\begin{algorithmic}[1]
\State Let $\alpha =0$ and $S = \arg\max_{\hat{S}: c(\hat{S}) = 0} f(\hat{S})$ \label{st:init}
\State Let $\opt = \max_{\hat{S} \subseteq \actions} (f(\hat{S}) - c(\hat{S}))$
\For{$j\in\actions$ with $c_j>0$}
\For{$k=0,\ldots,\lceil \log_{1/(1-\epsilon)}m \cdot 2^m \rceil $}
\State Let $\alpha_{j, k} = 1 - (1-\epsilon)^{k+1} \cdot  \frac{\opt}{ c_{j} + \opt}$
\State Let $S_{j,k} = \arg\max_{\hat{S}\subseteq \actions}  \left(f(\hat{S}) -\sum_{j'\in \hat{S}} \frac{c_{j'}}{\alpha_{j,k}}\right)$ 
\If{$(1-\alpha_{j, k}) f(S_{j,k}) > (1-\alpha)f(S)$}
\State $\alpha=\alpha_{j, k}$
\State $S= S_{j,k}$
\EndIf
\EndFor
\EndFor
\State \Return $\alpha,S$
\end{algorithmic}
\end{algorithm}

We are now ready to prove Theorem~\ref{thm:FPTAS}.

\begin{proof}[Proof of Theorem~\ref{thm:FPTAS}]
Let $\alpha^\star,S^\star$ be the best contract with its best response. Note that if $u_p(S^\star,\alpha^\star) \leq 0$, the claim holds trivially because Algorithm~\ref{alg:fptas} ensures that $u_p(S, \alpha) \geq 0$. Otherwise, $u_p(S^\star,\alpha^\star) = (1-\alpha^\star)f(S^\star) > 0$. It must then hold that $\alpha^\star < 1$ and $f(S^\star) > 0$. 
So, in particular, $S^\star \neq \emptyset$, so that $j^\star \in \arg\max_{j \in S^\star} c_j$ is well defined.  Also note that $\opt = \max_{\hat{S} \subseteq A} (f(\hat{S}) - c(\hat{S})) \geq f(S^\star) - c(S^\star) \geq (1-\alpha^\star) f(S^\star) >0$, because the agent's utility from $S^\star$ is non-negative. 

If $c_{j^\star}=0$, then the optimal contract is $\alpha=0$, which is the contract considered in Step~\ref{st:init} of the algorithm. 
Otherwise, $c_{j^\star} > 0$. Consider the iteration of Algorithm~\ref{alg:fptas} in which $j = j^\star$.
Let us denote $\alpha_{\min} =1- \frac{\opt}{ c_{j^\star} + \opt}$ and $\alpha_{\max} = 1- \frac{\opt}{m\cdot 2^m (c_{j^\star} + \opt )}$.
We claim that then there must be a value of  $k \in \{0,\ldots,\lceil \log_{1/(1-\epsilon)}m \cdot 2^m \rceil\}$ such that $1 - \alpha_{j, k}  \leq 1 - \alpha^\star \leq  \frac{1 - \alpha_{j, k}}{1-\epsilon}$.
Indeed, for $k = 0$ we have $\frac{1-\alpha_{j^\star,0}}{1-\epsilon} =  1- \alpha_{\min} \geq  1-\alpha^\star$, where the inequality follows by Lemma~\ref{lem:4m}, while for $k =\lceil \log_{1/(1-\epsilon)}m\cdot 2^m \rceil$ we have $1-\alpha^\star \geq 1- \alpha_{\max} \geq 1-\alpha_{j^\star,\lceil \log_{1/(1-\epsilon)}m\cdot 2^m \rceil}$, where the first inequality follows again by Lemma~\ref{lem:4m}. So there must be a $k \in \{0,\ldots,\lceil \log_{1/(1-\epsilon)}m \cdot 2^m \rceil\}$ with the desired properties.

We claim that for this choice of $j,k$, contract $\alpha_{j, k}$ provides a $(1-\epsilon)$-approximation to the optimal contract.
To see this, let $S$ be the choice of the agent under $\alpha_{j, k}$.
Using Proposition 3.1 in \cite{DuttingEFK21} (showing monotonicity of $f$ of the best response as a function of the contract $\alpha$), since $\alpha_{j, k} \geq \alpha^\star$, it must hold that $f(S) \geq f(S^\star)$. We thus obtain,
\[
(1-\alpha_{j, k}) f(S) \geq (1-\alpha_{j, k}) f(S^\star) \geq  (1-\epsilon)(1-\alpha^\star) f(S^\star) ,
\]
which completes the proof.
\end{proof}

The query complexity of Algorithm~\ref{alg:fptas} is upper bounded by $O(m \cdot \log_{1/(1-\epsilon)}m\cdot 2^m)$, which is $O(\frac{m^2}{\epsilon})$. 
We note that by using Proposition~\ref{prop:welfare-single}, one can tighten Lemma~\ref{lem:4m} for subadditive rewards, and improve the query complexity to $O(m \cdot \log_{1/(1-\epsilon)}m^2)$, which is $O\left(\frac{m\log m}{\epsilon}\right)$.

\subsection{Making Single-Agent Contracts Robust}\label{sec:making-single-agent-robust}

Algorithm~\ref{alg:robust-single-agent} runs, for every agent $i'$, the FPTAS from Section~\ref{sec:fptas-single-agent}, with $\epsilon = 1/2$, which returns a single-agent contract $\alpha'_i$ and best response set $S'_i$. 
The algorithm then considers an agent $i$ that maximizes $(1-\alpha'_i)f(S'_i)$.

A natural extension of this single-agent contract into a multi-agent contract $\contract$ would be to set $\alpha_i = \alpha'_i$ and $\alpha_{i'} = 0$ for all $i' \neq i$. Indeed, this ensures that $S_i = S'_i$ and $S_{i'} = \emptyset$ for all $i' \neq i$ is an equilibrium for $\contract$. However, this contract may admit additional equilibria.
In particular, agents $i' \neq i$ may take zero-cost actions, in which case $S_i = S'_i$ may no longer be a best response for agent $i$. 

To address this problem, depending on whether the total value of the zero-cost actions is high or not, the algorithm either returns $\contract = \vec{\epsilon}$ for some small enough $\epsilon > 0$, or further increases $\alpha_i$ to $\alpha_i = \frac{1+\alpha'_i}{2}$ while keeping $\alpha_{i'} = 0$ for all $i' \neq i$. The following lemma shows the robust approximation guarantee that the algorithm obtains this way. The proof once again relies on Lemma~\ref{lemma:anyequilibrium}. 

\begin{algorithm}
\caption{Robust Approximation of Single Agent Benchmark}\label{alg:robust-single-agent}
   \hspace*{\algorithmicindent} \textbf{Input:}  Costs $c_1,\ldots,c_m \in \reals_{\geq 0}$, value and demand oracle access to a submodular function $f:2^A \rightarrow \reals_{\geq 0}$ \\
    \hspace*{\algorithmicindent} \textbf{Output:}  A contract $\contract$ with a guarantee of principal's utility of $\Lambda$ for any equilibrium of $\contract$
\begin{algorithmic}[1]
\State Let $\contract' = \vec{0}$
\State Let $\actions^0 = \{j \in \actions \mid c_j = 0\}$ be the zero cost actions
\State Run Algorithm~\ref{alg:fptas} (with $\epsilon = \frac{1}{2}$) for each agent $i \in \agents$, to find $\alpha'_i$ and best response set $S'_i$
\State Let agent $i$ be the agent such that $(1-\alpha'_{i})f(S'_{i})$ is highest \label{step:best-agent}
\If{$f(A^0) > \frac{1}{3} (1-\alpha'_i)f(S'_i)$} 
\State Let $\epsilon = \frac{1}{2n}$
\State Let $\contract = \vec{\epsilon}$
\State Let $\Lambda = \frac{1}{4} f(A^0)$  
\Else
\State Let $\alpha_i = \frac{1+\alpha'_i}{2}$ and $\alpha_{i'} = 0$ for all $i' \neq i$
\State Let $\Lambda = \frac{1}{6} (1-\alpha'_i)f(S'_i)$  
\EndIf
\State \Return $\contract,\Lambda $
\end{algorithmic}
\end{algorithm}

\begin{lemma}\label{lem:single-agent-robust}
    Let $i, \alpha'_i, S'_i$ be defined as in Step~\ref{step:best-agent} of Algorithm~\ref{alg:robust-single-agent}, and $\contract,\Lambda $ be the returned contract and guarantee of Algorithm~\ref{alg:robust-single-agent}. Then, any $S$ that is a best response to $\contract$  guarantees the principal a utility of at least $\Lambda$, and it holds that
    \[
    \Lambda \geq 
    \frac{1}{12} (1-\alpha'_i) f(S'_i).
    \]
\end{lemma}

\begin{proof}
    We first consider the case where $f(\actions^0) > \frac{1}{3} (1-\alpha'_i) f(S'_i)$. 
    Note that $\actions^0$ is subset stable under $\contract = \vec{0}$.
    So, by Lemma~\ref{lemma:anyequilibrium}, any equilibrium $S$ with respect to $\contract = \vec{\epsilon}$ satisfies $f(S) \geq \frac{1}{2}f(\actions^0).$ We thus obtain
    \[
    \left(1-\sum_{i'} \alpha_{i'} \right) f(S) = (1-n\epsilon) f(S) = \frac{1}{2}f(S) \geq \frac{1}{4}  f(\actions^0) = \Lambda > \frac{1}{12} (1-\alpha'_i) f(S'_i).
    \]
    
    Otherwise, $\alpha_i = \frac{1 + \alpha'_i}{2}$ and $\alpha_{i'} = 0$ for $i' \neq i$. Note that $\alpha_i > \alpha'_i$ because for the $\alpha'_i$ returned by Algorithm~\ref{alg:fptas} it holds that $\alpha'_i < 1$. Consider any set of actions $S$ that is incentivized by $\contract$.
    We claim that $f(S) \geq \frac{1}{3} f(S'_i)$.
    This shows the claim because then
    \begin{align*}
    \bigg(1 - \sum_{i'} \alpha_{i'}\bigg) f(S) = (1-\alpha_i)f(S) &= \frac{1}{2}(1-\alpha'_i) f(S) \geq 
    \frac{1}{6}(1-\alpha'_i) f(S'_i) = \Lambda .
    \end{align*}
    
    It remains to show that $f(S) \geq \frac{1}{3} f(S'_i)$. We first use that under $\contract$ agent $i$ weakly prefers $S_i$ over $S'_i$. This shows that
    \begin{equation}
    \alpha_i f(S) - c(S_i) \geq \alpha_i f(S'_i \cup S_{-i}) - c(S'_i) \geq \alpha_i f(S'_i) - c(S'_i).
    \label{eq:1b}
    \end{equation}
    On the other hand, under $\contract'$ agent $i$ weakly prefers $S'_i$ over $S_i$ and thus
    \[
    \alpha_i' f(S'_i) - c(S'_i) \geq \alpha'_i f(S_i) - c(S_i),
    \]
    which implies 
    \begin{equation}
    \alpha'_i f(S'_i) - c(S'_i) \geq \alpha'_i f(S) - \alpha'_{i} f(S_{-i}) - c(S_i).\label{eq:2b}
    \end{equation}
    By summing up~\eqref{eq:1b} and~\eqref{eq:2b}, we get 
    \[
    (\alpha_i - \alpha'_i) f(S) \geq (\alpha_i - \alpha'_i) f(S'_i) - \alpha'_i f(S_{-i}),
    \]
    or equivalently 
    \[
    f(S) \geq f(S'_i) - \frac{\alpha'_i}{\alpha_i - \alpha'_i} f(S_{-i}).
    \]
    
    Note that since $\alpha_{i'} = 0$ for all $i' \neq i$, we must have $S_{-i} \subseteq \actions^0$. We thus have $f(S) \geq f(S'_i) - \frac{\alpha'_i}{\alpha_i-\alpha'_i} f(\actions^0).$
    Moreover, by the definition of $\alpha_i$, we have $\frac{\alpha'_i}{\alpha_i - \alpha'_i} = \frac{2 \alpha'_i}{1 - \alpha'_i} \leq \frac{2}{1 - \alpha'_i}$. Furthermore, $\frac{1}{1 - \alpha'_i} f(\actions^0) \leq \frac{1}{3} f(S'_i)$. Therefore,
    \[
    f(S) \geq \frac{1}{3} f(S'_i),
    \]    
    as claimed.
    \end{proof}

    We are now ready to prove Theorem~\ref{thm:robust-single-agent}.

    \begin{proof}[Proof of Theorem~\ref{thm:robust-single-agent}] The algorithm clearly runs in polynomial time using polynomially many value and demand queries.
    We claim that Algorithm~\ref{alg:robust-single-agent} 
    satisfies the assertion of the theorem.
    Let $\contract$ be the contract returned by Algorithm~\ref{alg:robust-single-agent}, and let $S$ be any equilibrium of $\contract$.

    Let $\contract^\star$ be any contract and any equilibrium $S^\star$ of $\contract^\star$ such that $\alpha^\star_i = 0$ and $S^\star_i = \emptyset$ for all but one $i \in \agents$. Then, by Theorem~\ref{thm:FPTAS},  $\alpha'_i$ and $S'_i$ as defined in Line~\ref{step:best-agent} of the algorithm, fulfill $(1-\alpha'_i)f(S'_i) \geq \frac{1}{2} (1- \alpha^\star_i) f(S^\star_i)$. So, by Lemma~\ref{lem:single-agent-robust}, we have $\Lambda \geq \frac{1}{24} \cdot (1-\sum_{i \in \agents} \alpha^\star_i)f(S^\star)$, as desired.
    \end{proof}

\section{Putting it all Together}\label{sec:putting-together}

In this section, we wrap up the proof of Main Theorem~\ref{mainthm:1}. So far, we presented two algorithms: Algorithm~\ref{alg:cap} gives a constant-factor approximation against the optimal contract in which no agent is large; Algorithm~\ref{alg:robust-single-agent} gives an approximation against the best single-agent contract. In this section, we show that restricting to these two cases incurs only a small loss.
So, taking the better of the two solutions will give us a constant-factor approximation against the overall optimal contract.

In more detail, the main remaining statement to be shown is the following.

\begin{theorem}[Reduction to no agent is large / single agent]
    \label{theorem:nolargeorsingleagent}
    Consider any submodular success probability function $f$. Let $\contract^\star$ be any contract and $S^\star$ be an equilibrium of $\contract^\star$. For any $0 \leq \phi \leq 1$, there exists a contract $\contract'$ and an equilibrium $S'$ of $\contract'$
    fulfilling
    $(1-\sum_i \alpha'_i) f(S') \geq \frac{\phi}{1120} \cdot (1-\sum_i \alpha^\star_i) f(S^\star)$ and
    \begin{enumerate}
        \item \textbf{(no large agent)} $f(S'_i) \leq \phi f(S')$ for all $i \in \agents$, or
        \item \textbf{(single agent)} 
        $\alpha'_i = 0$ and $S_i = \emptyset$ for all but one $i \in \agents$.
    \end{enumerate}
\end{theorem}

We prove Theorem~\ref{theorem:nolargeorsingleagent} below. 
Before, let's observe that with Theorem~\ref{theorem:nolargeorsingleagent} in hand, the proof of Main Theorem~\ref{mainthm:1} becomes straightforward.

\begin{algorithm}
\caption{Robust Approximation for the Optimal Contract}\label{alg:meta}
   \hspace*{\algorithmicindent} \textbf{Input:}  Costs $c_1,\ldots,c_m \in \reals_{\geq 0}$, value and demand oracle access to a submodular function $f:2^A \rightarrow \reals_{\geq 0}$ \\
    \hspace*{\algorithmicindent} \textbf{Output:}  A contract $\contract$ that approximates the optimal contract
\begin{algorithmic}[1]
\State Run Algorithm~\ref{alg:cap} on $c_1,\ldots,c_m$, $f$, with parameters $\epsilon=\frac{1}{8n}$, $\rho=\frac{1}{6}$ and receive a contract $\contract_{\textsf{multi-agent}}$ with guarantee $\Lambda_{\textsf{multi-agent}}$
\State Run Algorithm~\ref{alg:robust-single-agent} on $c_1,\ldots,c_m$, $f$, and receive a contract $\contract_{\textsf{single-agent}}$ with guarantee $\Lambda_{\textsf{single-agent}}$
\If{$\Lambda_{\textsf{multi-agent}} > \Lambda_{\textsf{single-agent}} $}
\State $\contract = \contract_{\textsf{multi-agent}}$
\Else
\State $\contract = \contract_{\textsf{single-agent}}$
\EndIf
\State \Return $\contract$
\end{algorithmic}
\end{algorithm}

\begin{proof}[Proof of Main Theorem~\ref{mainthm:1}]
We show that Algorithm~\ref{alg:meta} 
satisfies the assertion of the theorem.  

By construction, the returned contract has the property that every equilibrium has principal utility at least $\Lambda = \max\{\Lambda_{\textsf{multi-agent}}, \Lambda_{\textsf{single-agent}}\}$. We still have to show that $\Lambda = \Omega(1) u_p(\contract^\star, S^\star)$.

To this end, let $\phi$ be the constant from Theorem~\ref{thm:noagentislarge}. By Theorem~\ref{theorem:nolargeorsingleagent}, there exists a contract $\contract'$ with an equilibrium $S'$ fulfilling $(1-\sum_i \alpha'_i) f(S') \geq \frac{\phi}{1120} \cdot (1-\sum_i \alpha^\star_i) f(S^\star)$ and (1) $f(S'_i) \leq \phi f(S')$ for all $i \in \agents$, or (2) $\alpha'_i = 0$ and $S_i = \emptyset$ for all but one $i \in \agents$.

We will distinguish these two cases. In the first case, Theorem~\ref{thm:noagentislarge} ensures that $\Lambda_{\textsf{multi-agent}} \geq \kappa (1-\sum_i \alpha'_i) f(S') \geq \kappa \frac{\phi}{1120} \cdot (1-\sum_i \alpha^\star_i) f(S^\star)$; in the second case, Theorem~\ref{thm:robust-single-agent} ensures that $\Lambda_{\textsf{single-agent}} \geq \frac{1}{12} (1-\sum_i \alpha'_i) f(S') \geq \frac{1}{12} \frac{\phi}{1120} \cdot (1-\sum_i \alpha^\star_i) f(S^\star)$.
\end{proof}

\paragraph{Reduction to no agent is large or single agent.}

The key difficulty in proving Theorem~\ref{theorem:nolargeorsingleagent} is the lack of montonicity demonstrated in Example~\ref{example:more-or-less}. 
That is, when removing all agents except for $i$, this agent might do less. To mitigate this effect, we would like to increase $\alpha_i$. However, this is problematic if $\alpha_i$ is already close to $1$. Therefore, our key lemma for showing Theorem~\ref{theorem:nolargeorsingleagent} establishes that it is sufficient to consider contracts in which all of the $\alpha_i$ are bounded away from $1$ or all but one of them are $0$.

\begin{lemma}\label{lem:red2}
    Consider any submodular success probability function $f$. Let $\contract^\star$ be any contract and let $S^\star$ be any equilibrium of $\contract^\star$. Then there exist a contract $\contract'$ and an equilibrium $S'$ of $\contract'$ fulfilling
    $(1-\sum_i \alpha'_i) f(S') \geq \frac{1}{20} \cdot (1-\sum_i \alpha^\star_i) f(S^\star)$ and
    \begin{enumerate}
        \item $\alpha'_i \leq \frac{3}{4}$ for all $i \in \agents$, or
        \item $\alpha'_i = 0$ and $S_i = \emptyset$ for all but one $i \in \agents$.
    \end{enumerate}
\end{lemma}

\begin{proof}
    Consider some contract $\contract^\star$ and some equilibrium $S^\star$ of $\contract^\star$. 
    Clearly, we can focus on the case where $\sum_{i} \alpha^\star_i < 1$. Note that in this case there is nothing to show when $\alpha^\star_i \leq 3/4$ for all $i$. Moreover, there can be at most one agent $i$ such that $\alpha^\star_i > 3/4$. Let $i$ denote the index of this agent.
        
    \medskip
    
    {\bf Case 1:} Let's first consider the case where $(1-\alpha_i^\star) f(S_i^\star) \geq 4 f(S_{-i}^\star)$. Then, because $\alpha^\star_i > 3/4$, we have $f(S_i^\star) \geq 16 f(S_{-i}^\star)$. 
    By subadditivity of $f$, this implies that $f(S_i^\star) \geq 16 (f(S^\star) - f(S_i^\star))$, or equivalently, $f(S_i^\star) \geq \frac{16}{17} f(S^\star)$.  
    
    Consider contract $\contract'$ with $\alpha'_i = \frac{1 + \alpha_i^\star}{2}$ and $\alpha'_{i'} = 0$ for $i' \neq i$. Note that $\alpha'_i > \alpha^\star_i$. For agents $i' \neq i$ it's a best response to play $S'_i = \emptyset$ (independent of what agent $i$ does). Let $S'_i$ be any best response by agent $i$ (when $S'_{i'} = \emptyset$ for all $i' \neq i$).
    We claim that $f(S'_i) \geq \frac{1}{2} f(S_i^\star)$.
    This shows the claim because then
    \begin{align*}
    \bigg(1 - \sum_{i'} \alpha'_{i'}\bigg) f(S') = (1-\alpha_i')f(S'_i) &= \frac{1}{2}(1-\alpha_i^\star) f(S'_i)\\ 
    &\geq 
    \frac{1}{4}(1-\alpha_i^\star) f(S_i^\star) \geq \frac{1}{4} \bigg(1 - \sum_{i'} \alpha^\star_{i'}\bigg) f(S^\star_i) \geq \frac{4}{17} \bigg(1 - \sum_{i'} \alpha^\star_{i'}\bigg) f(S^\star),
    \end{align*}
    where the second equality follows by the definition of $\alpha_i'$. Note that the sum is over {\em all} $i' \in [n]$.
    
    Towards showing that $f(S') \geq \frac{1}{2} f(S^\star_i)$, first observe that under $\contract'$ agent $i$ weakly prefers $S'_i$ over $S^\star_i$. Since $S'_{i'} = \emptyset$ for all $i' \neq i$ this shows that
    \begin{equation}
    \alpha_i' f(S'_i) - c(S_i') \geq \alpha_i' f(S_i^\star) - c(S_i^\star).
    \label{eq:1}
    \end{equation}
    On the other hand, under $\contract^\star$ agent $i$ weakly prefers $S^\star_i$ over $S'_i$ and thus, given the choices $S^\star_{-i}$ of the agents other than $i$, it must hold that
    \[
    \alpha_i^\star f(S_i^\star \mid S_{-i}^\star) - c(S_i^\star) \geq \alpha_i^\star f(S_i' \mid S_{-i}^\star) - c(S_i').
    \]
    By subadditivity of $f$, we have $f(S_i^\star \mid S_{-i}^\star) = f(S_i \cup S_{-i}) - f(S_{-i}) \leq f(S^\star_i)$. By monotonicity of $f$, we have $f(S'_i \mid S^\star_{-i}) = f(S'_i \cup S^\star_{-i}) - f(S^\star_{-i}) \geq f(S'_i) - f(S^\star_{-i})$. We thus obtain
    \begin{align}
    \alpha_i^\star f(S_i^\star) - c(S_i^\star) \geq \alpha_i^\star f(S_i') - \alpha_i^\star f(S_{-i}^\star) - c(S_i').
    \label{eq:2}
    \end{align}
    By summing up~\eqref{eq:1} and~\eqref{eq:2}, we get 
    \[
    (\alpha_i' - \alpha_i^\star) f(S'_i) \geq (\alpha_i' - \alpha_i^\star) f(S_i^\star) - \alpha_i^\star f(S_{-i}^\star),
    \]
    or equivalently (using $\alpha_i' > \alpha_i^\star$)
    \[
    f(S'_i) \geq f(S_i^\star) - \frac{\alpha_i^\star}{\alpha_i' - \alpha_i^\star} f(S_{-i}^\star).
    \]
    
    By the definition of $\alpha_i'$, we have $\frac{\alpha_i^\star}{\alpha_i' - \alpha_i^\star} = \frac{2 \alpha_i^\star}{1 - \alpha_i^\star} \leq \frac{2}{1 - \alpha_i^\star}$. Applying the condition of Case 1, we furthermore have $\frac{1}{1 - \alpha_i^\star} f(S_{-i}^\star) \leq \frac{1}{4} f(S_i^\star)$. So, $\frac{\alpha_i^\star}{\alpha_i' - \alpha_i^\star} f(S_{-i}^\star) \leq \frac{1}{2} f(S_i^\star)$ and therefore $f(S'_i) \geq \frac{1}{2}f(S^\star_i)$. This concludes the argument for this case.

    \medskip
    
    {\bf Case 2:} Next consider the case where $(1-\alpha^\star_i)f(S^\star_i) < 4 f(S^\star_{-i})$. In this case we have
    \begin{align}
        \bigg(1 - \sum_{i'} \alpha^\star_{i'}\bigg) f(S^\star) &= \bigg(1 - \alpha^\star_i - \sum_{i'\neq i} \alpha^\star_{i'}\bigg) \left(f(S_i^\star \mid S_{-i}^\star) + f(S_{-i}^\star)\right) \notag\\
        &\leq (1 - \alpha^\star_i) f(S_i^\star \mid S_{-i}^\star) + f(S_{-i}^\star) \leq (1 - \alpha^\star_i) f(S_i^\star) + f(S_{-i}^\star) \leq 5 f(S_{-i}^\star), \label{eq:lbonf}
    \end{align}
    where the last inequality follows by the condition of Case 2, and the one before follows by subadditivity of $f$.
    
    In order to define a contract without agent $i$, we will use Lemma~\ref{lemma:anyequilibrium}. To this end, we will show that the set $S^\star_{-i} = S^\star \setminus S^\star_i$ is subset stable for contract $\contract^\star$ (and in fact also for contract $\vec{\beta}$ which is equal to $\contract^\star$ in all coordinates but coordinate $i$, where $\beta_i = 0$). 
    
    To see this, observe that for all $i' \neq i$ and all $S_{i'}' \subseteq S_{i'}^\star$ we have
    \[
    \alpha^\star_{i'} f(S^\star) - c(S_{i'}^\star) \geq \alpha^\star_{i'} f(S_{i'}', S_{-i'}^\star) - c(S_{i'}'),
    \]
    or equivalently
    \[
    \alpha^\star_{i'} f(S_{i'}^\star \setminus S_{i'}' \mid S_{i'}' \cup S_{-i'}^\star) \geq c(S_{i'}^\star) - c(S_{i'}').
    \]
    By submodularity of $f$, we have $f(S_{i'}^\star \setminus S_{i'}' \mid S_{i'}' \cup S_{-i'}^\star) \leq f(S_{i'}^\star \setminus S_{i'}' \mid S_{i'}' \cup (S_{-i'}^\star\setminus S^\star_i))$. Moreover, since $S'_{i'} \subseteq S^\star_{i'}$, it holds that $f(S_{i'}^\star \setminus S_{i'}' \mid S_{i'}' \cup (S_{-i'}^\star\setminus S^\star_i)) = f(S_{i'}^\star \cup (S_{-i'}^\star\setminus S^\star_i) ) - f(S_{i'}' \cup (S_{-i'}^\star\setminus S^\star_i))$. We thus obtain
    \[
    \alpha^\star_{i'} f(S_{i'}^\star, S^\star_{-i'} \setminus S_i^\star) - c(S_{i'}^\star) \geq \alpha^\star_{i'} f(S_{i'}', S_{-i'}^\star \setminus S_i^\star) - c(S_{i'}'). 
    \]
    
    Consider contract $\vec{\beta}$ with $\beta_{i'} = \alpha_{i'}^\star$ for $i' \neq i$ and $\beta_i = 0$. Since $\alpha^\star_i > 3/4$, we can choose $\epsilon > 0$ small enough so that $\sum_{i' \neq i} 2\alpha^\star_{i'} + n\epsilon \leq 1/2$. 
    Consider any equilibrium $\Seq{}$ with respect to $\contract' = 2 \vec{\beta} + \vec{\epsilon}$. 
    Then, by Lemma~\ref{lemma:anyequilibrium}, we have $f(\Seq{}) \geq \frac{1}{2} f(S^\star_{-i})$. So, the principal utility of contract $2 \vec{\beta} + \vec{\epsilon}$ in equilibrium $\Seq{}$ is 
    \[
    \left(1 - \sum_{i'\neq i} 2 \alpha^\star_{i'} - n \epsilon\right) f(\Seq{}) \geq \frac{1}{4} f(S_{-i}^\star) \geq \frac{1}{20} \left(1 - \sum_{i'} \alpha^\star_{i'}\right) f(S^\star),
    \]
    where for the first inequality we used that $f(\Seq{}) \geq \frac{1}{2}f(S^\star_{-i})$ and $\sum_{i' \neq i} 2\alpha^\star_{i'} + n\epsilon \leq 1/2$. The second inequality holds by Inequality~\eqref{eq:lbonf}. As we have $\sum_{i'} \alpha'_{i'} = \sum_{i' \neq i} 2\alpha^\star_{i'} + n\epsilon \leq 1/2$, in particular, we also have $\alpha'_{i'} \leq 3/4$ for all $i'$. This completes the argument for this case.
\end{proof}

We are now ready to prove Theorem~\ref{theorem:nolargeorsingleagent}.

\begin{proof}[Proof of Theorem~\ref{theorem:nolargeorsingleagent}]
    By Lemma~\ref{lem:red2}, there is a contract $\alpha'$ incentivizing $S'$ such that $(1-\sum_i \alpha'_i) f(S') \geq \frac{1}{20} \cdot (1-\sum_i \alpha^\star_i) f(S^\star)$ and $\alpha'_i \leq \frac{3}{4}$ for all $i \in \agents$, or $\alpha'_i = 0$ and $S'_i = \emptyset$ for all but one $i \in \agents$.

    In the latter case, there is nothing to be shown. We are done as well if $f(S_i') \leq \phi f(S')$ for all $i \in \agents$. Otherwise, there is an agent $i$ such that $f(S_i') > \phi f(S')$ and $\alpha_i' \leq \frac{3}{4}$. Now consider the contract $\contract''$ that sets $\alpha_i'' = \frac{7}{8}$ for this agent $i$ and $\alpha_{i'}'' = 0$ for all agents $i' \neq i$. For agents $i' \neq i$ it is a best response to play $S''_{i'} = \emptyset$. 
    Let $S_i''$ be any best response by agent $i$ under this choice of $S''_{i'}$ for $i' \neq i$.

    We now have
    \[
    \frac{7}{8} f(S_i'') \geq \alpha_i'' f(S_i'') - c(S_i'') \geq \alpha_i'' f(S_i') - c(S_i') \geq \frac{1}{8} f(S_i') + \alpha_i' f(S_i') - c(S_i') \geq \frac{1}{8} f(S_i').
    \]

    So,
    \[
    (1 - \alpha_i'') f(S_i'') = \frac{1}{8} f(S_i'') \geq \frac{1}{56} f(S_i') \geq \frac{\phi}{56} f(S') \geq \frac{\phi}{56} (1-\sum_{i'} \alpha'_{i'}) f(S') \geq \frac{\phi}{1120} (1-\sum_{i'} \alpha^\star_{i'}) f(S^\star),
    \]
    as claimed.
\end{proof}

\section{No PTAS for Multiple Agents, Not Even with Demand Queries}
\label{sec:noptas}

In this section, we show that even in the case of binary actions (the model studied in \cite{DuettingEFK23}), there is an absolute constant $\eta<1$ for which no algorithm can guarantee an approximation of $\eta$ with a polynomial number of value and demand queries.
This result strengthens a result by \cite{EzraFS24}, who showed that no poly-time algorithm can get $0.7$ approximation with access to value oracle alone.

\begin{theorem}\label{thm:upper}
    There exists a constant $\eta<1$ for which any algorithm that uses a sub-exponential number of value and demand queries to the submodular success probability function returns a $\eta$ -approximate optimal contract with an exponentially small probability in $n$.  
\end{theorem}

Our proof relies on the construction of a family of functions $\{f_T:2^\actions\rightarrow R_{\geq 0}\}_T$, defined using the following function, parameterized by $k$:

\[
f(x,y) = \begin{cases}
    \sqrt{k} & \text{ if $x+y \geq k$} \\
    \sqrt{x+y} & \text{ if $x+y < k$ and $y \leq \frac{k}{2}$} \\
    \sqrt{x + \frac{k}{2}}  +  \frac{y- \frac{k}{2}}{\sqrt{k} + \sqrt{x + \frac{k}{2}}} & \text{ otherwise,}
\end{cases}
\]
where
$$ f_T(S)= f(|S\setminus T|, |S\cap T|) \mbox{ with parameter } k=|T|.$$

The construction of the hard example is different in nature than the constructions given in \cite{DuettingEFK23,EzraFS24} for XOS functions. The main idea in  \cite{DuettingEFK23} is to hide a single good set of agents with a higher value which cannot be learned efficiently using value and demand queries. This increased value decreases the payments to all agents, leading to an attractive set of agents. However, in order to make all its subsets non-attractive, the construction heavily uses the non-monotonicity of the marginals in XOS functions. Therefore, this approach cannot be extended to submodular functions.
The hardness result of \cite{EzraFS24} for XOS functions in the case of value queries also relies on hiding a small set of good agents with high value, which can be incentivized by relatively low payments. However, since the attractive set has a much higher value than a random set of the same size, this type of hidden set can be learned relatively easily by demand queries.
Therefore, establishing a hardness result in our case requires new ideas.

Our construction is based on creating a good set of agents whose value is indistinguishable from other sets of the same size (so demand queries would not be useful in identifying the hidden set). 
Thus, the only way to make the set more attractive is by making it cheaper to incentivize. This is done by decreasing the value of its subsets (instead of increasing the value of the set itself as was done in the former constructions). This decrease in value makes the marginal contribution of every agent higher,  which in turn decreases the payments necessary to incentivize them, while keeping the attractive set (and all its subsets) hidden, as demand oracles are better in finding sets with higher values than sets with lower values.

\begin{claim}
    For every set $T\subset A$ of even size, the function $f_T$ is monotone and submodular.
    \label{cl:submodularity}
\end{claim}
\begin{proof}

Let 
\[
f_x(x,y) = f(x+1,y)-f(x,y) = \begin{cases}
    0 & \text{ if $x+y \geq k$} \\
        \frac{1}{\sqrt{x+y+1} + \sqrt{x+y}} & \text{ if $x+y < k$ and $y \leq \frac{k}{2}$} \\
    \frac{1}{\sqrt{x +1+ \frac{k}{2}}+\sqrt{x + \frac{k}{2}}}  +  \frac{y- \frac{k}{2}}{\sqrt{k} + \sqrt{x +1+ \frac{k}{2}}} -  \frac{y- \frac{k}{2}}{\sqrt{k} + \sqrt{x + \frac{k}{2}}}  & \text{ otherwise}
\end{cases}
\]
and 
\[
f_y(x,y) = f(x,y+1)-f(x,y) = \begin{cases}
    0 & \text{ if $x+y \geq k$} \\
    \frac{1}{\sqrt{x+y+1} + \sqrt{x+y}} & \text{ if $x+y < k$ and $y < \frac{k}{2}$} \\
     \frac{1}{\sqrt{k} + \sqrt{x + \frac{k}{2}}}  & \text{ otherwise.}
\end{cases}
\]
To show that $f_T$ is monotone, it is equivalent to prove that $f_x(x,y),f_y(x,y)\geq 0$ for every $x,y$.
The only case that is not trivial is for $f_x$, where we are in the third case for which \begin{eqnarray*}
f_x(x,y) & = &   \frac{1}{\sqrt{x +1+ \frac{k}{2}}+\sqrt{x + \frac{k}{2}}}  + \underbrace{  \frac{y- \frac{k}{2}}{\sqrt{k} + \sqrt{x +1+ \frac{k}{2}}} -  \frac{y- \frac{k}{2}}{\sqrt{k} + \sqrt{x + \frac{k}{2}}}}_{\leq  0} \\ & \geq &  
\frac{1}{\sqrt{x +1+ \frac{k}{2}}+\sqrt{x + \frac{k}{2}}}  +  \frac{ \frac{k}{2}-x}{\sqrt{k} + \sqrt{x +1+ \frac{k}{2}}} -  \frac{\frac{k}{2}-x}{\sqrt{k} + \sqrt{x + \frac{k}{2}}} \\ & = & \frac{1}{\sqrt{k} + \sqrt{1 + \frac{k}{2} + x}} \geq 0,
\end{eqnarray*}
where the first inequality holds since in this case, $x+y \leq k$.

To show submodularity, it is sufficient to show that the derivative of all cases of $f_x$ and $f_y$ with respect to both $x,y$ are non-positive, and that $f_x$ and $f_y$ decrease in the transitions between the cases of the function.
The only non-trivial case for the derivative is again the third case of $f_x$ for which  \begin{eqnarray*}
\frac{f_x(x,y)}{\delta x} & = & 
\underbrace{\frac{-1}{\sqrt{k/2 + x} \sqrt{2 + k + 2 x} (\sqrt{k + 2 x} + \sqrt{2 + k + 2 x})}}_{\textit{increasing in } x} 
\\ & + & \underbrace{\frac{y-\frac{k}{2} }{2 \sqrt{\frac{k}{2} + x} (\sqrt{k} + \sqrt{\frac{k}{2} + x})^2} - \frac{y-\frac{k}{2} }{2 \sqrt{1 + \frac{k}{2} + x} (\sqrt{k} + \sqrt{1 + \frac{k}{2} + x})^2}}_{\textit{decreasing in } x \textit{ and incresing in } y}\\ & \leq & 
-\frac{1}{\sqrt{k} \sqrt{2k} (\sqrt{2k } + \sqrt{2k})} +\frac{\frac{k}{2} }{2 \sqrt{\frac{k}{2} } (\sqrt{k} + \sqrt{\frac{k}{2} })^2} - \frac{\frac{k}{2} }{2 \sqrt{1 + \frac{k}{2} } (\sqrt{k} + \sqrt{1 + \frac{k}{2} })^2} \leq 0,
\end{eqnarray*}
where for the first inequality we use that $x\in[0,\frac{k}{2}]$ and that $y\leq k$, and the last inequality holds for every $k>0$.

We also need to show that in the transitions between the cases, the submodularity is maintained. That is, we need to show that (1) $f_x(x,k/2) \geq f_x(x,k/2+1) $ for all $x \leq k/2$, and (2) that $f_y(x,k/2-1) \geq f_y(x,k/2)$ for all $x\leq k/2$.
The first inequality holds since 
\begin{eqnarray*}
f_x(x,k/2)  & = & \frac{1}{\sqrt{x +1+ \frac{k}{2}}+\sqrt{x + \frac{k}{2}}} \\ &  \geq &  \frac{1}{\sqrt{x +1+ \frac{k}{2}}+\sqrt{x + \frac{k}{2}}}  +  \underbrace{\frac{1}{\sqrt{k} + \sqrt{x +1+ \frac{k}{2}}} -  \frac{1}{\sqrt{k} + \sqrt{x + \frac{k}{2}}}}_{ \leq 0 } =  f_x(x,k/2+1), 
\end{eqnarray*}
and the second inequality holds since 
\begin{eqnarray*}
f_y(x,k/2-1) & = & \frac{1}{\sqrt{x+\frac{k}{2}} + \sqrt{x+\frac{k}{2}-1}} \geq  \frac{1}{\sqrt{k} + \sqrt{x + \frac{k}{2}}} =  f_y(x,k/2), 
\end{eqnarray*}
where the inequality is since $x \leq \frac{k}{2}$.
\end{proof}

\bigskip

Consider an instance defined by $f_T$ for some set $T$ where $
c_i = \frac{1}{8k^{3/2}}$ for every $i\in\actions$.
Let $g_T$ be the corresponding principal's utility function with respect to $f_T$.

\begin{lemma}\label{lem:sub-opt}
    For any set $X$ with $\lvert X \cap T \rvert \leq \frac{k}{2}$, we have $g_T(X) < (0.97 + \frac{1}{3 \sqrt{k}}) \cdot \max_{X^\ast} g_T(X^\ast)$.
\end{lemma}

\begin{proof}
We first show that $\max_{X^\ast} g_T(X^\ast) > 0.78 \sqrt{k}$. 
In particular, we show that $g_T(T) > 0.78 \sqrt{k}$.
First observe that $f_T(T) = \sqrt{k}$. Furthermore, for every $i \in T$, we have
\[
f_T(T \setminus \{i\}) = \sqrt{\frac{k}{2}} + \left(\sqrt{k} - \sqrt{\frac{k}{2}}\right) \frac{k - 1 - \frac{k}{2}}{\frac{k}{2}} = \sqrt{\frac{k}{2}} + \sqrt{k} \left(1 - \frac{1}{\sqrt{2}}\right) \left(1 - \frac{2}{k}\right).
\]
So the marginal of every element is given by
\[
f_T(i \mid T \setminus \{i\}) =  f_T(T) - f_T(T \setminus \{i\}) = \left(1 - \frac{1}{\sqrt{2}}\right) \frac{2}{\sqrt{k}}  = \frac{2-\sqrt{2}}{\sqrt{k}},
\]
implying 
\[
\frac{c_i}{f_T(i \mid T \setminus \{i\})} = \frac{1}{8 k^{3/2}} \frac{\sqrt{k}}{2 - \sqrt{2}} = \frac{1}{k} \frac{1}{8(2 - \sqrt{2})}
\]
and so
\[
1 - \sum_{i \in T} \frac{c_i}{f_T(i \mid T \setminus \{i\})} = 1 - \frac{1}{8(2 - \sqrt{2})} > 0.78.
\]
This implies that
\[
g_T(T) = \left( 1 - \sum_{i \in T} \frac{c_i}{f_T(i \mid T \setminus \{i\})}\right) f_T(T) > 0.78 \sqrt{k}.
\]
Given that $g_T(T) > 0.78 \sqrt{k}$, to prove the lemma it suffices to show that $g_T(X) \leq \frac{3}{4} \sqrt{k} + \frac{1}{4}$. 

If $\lvert X \cap T \rvert \leq \frac{k}{2}$ and $\lvert X \rvert \leq k$, we have for all $i \in X$
\[
\frac{c_i}{f_T(i \mid X \setminus \{i\})} = \frac{c_i}{\sqrt{\lvert X \rvert} - \sqrt{\lvert X \rvert - 1}} = c_i (\sqrt{\lvert X \rvert} + \sqrt{\lvert X \rvert - 1}) \geq 2 c_i \sqrt{\lvert X \rvert - 1} \geq 2 c_i (\sqrt{\lvert X \rvert} - 1).
\]
Therefore,
\[
g_T(X) = \left( 1 - \sum_{i \in X} \frac{c_i}{f_T(i \mid X \setminus \{i\})}\right) f_T(X) \leq (1 - \lvert X \rvert 2 c_i (\sqrt{\lvert X \rvert} - 1)) \sqrt{\lvert X \rvert} = \sqrt{\lvert X \rvert} - \frac{\lvert X \rvert^2}{4 k^{3/2}} + \frac{\lvert X \rvert^{3/2}}{4 k^{3/2}} \leq \frac{3}{4} \sqrt{k} + \frac{1}{4},
\]
as desired, where we use that $x \mapsto \sqrt{x} - \frac{x^2}{4 k^{3/2}}$ is non-decreasing for $x \in [0, k]$.

Furthermore, if $\lvert X \rvert > k$, then $f_T(i \mid X \setminus \{i\}) = 0$ and so $g_T(X) = -\infty$.
\end{proof}
We next show how to implement a demand query for $f_T$ using access to value queries for $f_T$.
\begin{lemma}
\label{lem:demand}
    For every set $T \subseteq \actions $ of even size $4 \leq k<|\actions|$,  
    answering a demand query for $f_T$ can be done with polynomially many value queries to $f_T$ (without knowing the set $T$).
\end{lemma}

\begin{proof}
We present an algorithm that answers a demand set using poly many value queries. 

Let $p_1,\ldots,p_n$ be the 
price vector whose demand is desired; that is, 
we want to find a set $S$ that maximizes  $f_T(S) -\sum_{j\in S} p_j$. Without loss of generality assume that $p_1 \leq p_2\leq \ldots \leq p_n$.

We first observe that the algorithm can learn $k$ (i.e., the size of $T$), by the single value query $f_T(\actions)$. Indeed, since $T \subseteq \actions$, $f_T(\actions) = \sqrt{k}$. 

Next, the algorithm issues value queries of the form $f_T([i])$ for all $i\in [k]$.
If $f_T([i]) = \sqrt{i}$ for all $i\in [k]$, then the algorithm returns the set $S$ that maximizes $f_T(S) -\sum_{j\in S} p_j$ among all these sets. This must be a demand set since $[i]$ is the cheapest set of size $i$, no set of size $i$ has value greater than $\sqrt{i}$, and there is always a demand set of size at most $k$.

Otherwise, there exists $i \in [k]$ such that $f_T([i]) \neq \sqrt{i}$. 
We show that in this case the algorithm can find the set $T$ with poly-many additional value queries. We then show that if the algorithm knows $T$ then it can find a demand set with poly-many value queries.

Let $i^\star$ be the minimal $i$ such that $f_T([i]) \neq \sqrt{i}$. By the definition of $f_T$, we get that $i^\star<k$ and $|[i^\star] \cap T |> \frac{k}{2}$.
Since $[i^\star-1] =\sqrt{i^\star-1}$ and that $i^\star-1 <k$, it must be that $|[i^\star-1] \cap T| \leq \frac{k}{2} $, which implies that $i^\star \in T$.

Then, for each  $j \notin [i^\star]$ the algorithm issues the query $f_T([i^\star-1]\cup\{j\})$. If for $j$ it holds that $ f_T([i^\star-1]\cup\{j\}) = f_T([i^\star]) $, it means that either $j\in T$, otherwise, $f_T([i^\star-1]\cup\{j\})  = \sqrt{i^\star}$ and $j \notin T$.
Then, let $j^\star$ be an arbitrary action in $i \setminus [i^\star]$ (such exists since $i^\star <k$). The for each $j\in [i^\star]$ the algorithm issues the query  $ f_T([i^\star] \cup \{j^\star\} \setminus \{j\})  $, for which   if $ f_T([i^\star] \cup \{j^\star\} \setminus \{j\})  = f_T([i^\star])$ then $j \in T $ iff $j^\star \in T$. Since whether $j^\star$ in $T$ is already known, it means that we know for each $i\in \actions$ whether it is in $T$ or not. 

Once the identity of $T$ is known, one can answer a demand query as follows.
Go over all options for $x,y$ satisfying $0 \leq x+y \leq k$, and consider the set $S$ composed of the $x$ cheapest elements in $\actions \setminus T$, and the $y$ cheapest elements in $T$. The set $S$ among these sets that maximizes $f_T(S) -\sum_{j\in S} p_j$ is a demand.

This algorithm returns a demand set in polynomial time using $O(|\actions|)$ value queries, completing the proof of the lemma.
\end{proof}

We are now ready to present the proof of Theorem~\ref{thm:upper}.

\begin{proof}[Proof of Theorem~\ref{thm:upper}]
To prove the theorem, it suffices to show that the statement holds with respect to a deterministic algorithm against a randomized input, where the set $T$ that defines the function $f_T$ is chosen uniformly at random from all sets of size $k = 2 \cdot  \lfloor \frac{n}{10}\rfloor$. 
(Note that restricting the sets $T$ that the adversary chooses from only strengthens our statement.)

By Lemma~\ref{lem:demand}, it suffices to consider a deterministic algorithm that issues only value queries, and we may assume without loss of generality that the returned set is queried (otherwise, just add this query at the end of the algorithm).

By Lemma~\ref{lem:sub-opt}, for large enough $k$, in order to get better than $\eta = 0.98$ approximation, for large enough $n$, the algorithm must ask a value query for a set $S$ for which  $|S| \leq k $ and $|S \cap T| > \frac{k}{2}$. We call such a query a {\em good} query.

Since the probability of a query $S$ to be good is increasing in the size of the set $S$ (up to size $k$), the probability that a single query is successful is upper bounded by the probability that a hypergeometric random variable with parameters $(n,k,k)$ is greater than $k/2$. For our choice of $k$ (namely, $k = 2 \cdot  \lfloor \frac{n}{10}\rfloor$), this probability is exponentially small in $n$.

We note that, since we bound the probability that the algorithm issues a good query, it is without loss of generality to assume that the queries are non-adaptive. Indeed, at any point in time, it can be assumed that none of the previous queries were good, and for queries that are not good, no information beyond this fact is learned.

In summary, every query is good with exponentially small probability, thus, by the union bound, the probability of issuing a good query is exponentially small for a sub-exponential number of queries.
This concludes the proof.
\end{proof}

We remark that our proof applies to submodular success probability functions, but does not rule out the existence of a PTAS (or FPTAS) for the special case of gross substitutes functions. Indeed, the following observation shows that our construction does not satisfy the gross substitutes condition.

\begin{observation}
For every $T\subset \actions$ of even size $k\geq 4$, the function $f_T$ is not gross substitutes.  
\end{observation}

\begin{proof}
According to the triplet condition in \cite{ReijniersePG02}, a function $g:2^\actions\rightarrow \reals_{\geq 0}$ is gross substitutes if and only if it is submodular and for any set $S$, and any three elements $\{i,j,k\} \not\in S$, it holds that
\[
\max\{g(i,k \mid S) + g(j \mid S), g(j,k \mid S) + g(i \mid S)\} \geq g(i,j \mid S) + g(k \mid S).
\]
Consider two arbitrary elements $i_1,i_2\in T$, and arbitrary element $j\not\in T$, and let $S \subseteq T \setminus \{i_1,i_2\}$ be an arbitrary set of size $k/2$.

The triplet condition requires that
\[
\max\{f_T(i_1,j \mid S) + f_T(i_2 \mid S), f_T(i_2,j \mid S) + f_T(i_1 \mid S)\} \geq f_T(i_1,i_2 \mid S) + f_T(j \mid S),
\]
which is equivalent to
\[
   f(1, \frac{k}{2}+1) + f(0, \frac{k}{2} + 1) \geq f(1, \frac{k}{2}) + f(0, \frac{k}{2} + 2).
\]
However,
\begin{eqnarray*}
f(1, \frac{k}{2}+1) + f(0, \frac{k}{2} + 1) & = & 
\sqrt{\frac{k}{2}+1} + \frac{1}{\sqrt{k} + \sqrt{\frac{k}{2}+1}}+  \sqrt{\frac{k}{2}} + \frac{1}{\sqrt{k} + \sqrt{\frac{k}{2}}}\\ 
& < & 
\sqrt{\frac{k}{2}+1} + \sqrt{\frac{k}{2}} + \frac{2}{\sqrt{k} + \sqrt{\frac{k}{2}}} \\
& = & 
 f(1, \frac{k}{2}) + f(0, \frac{k}{2} + 2), \end{eqnarray*}
 which concludes the proof.
\end{proof}

\section{Guarantees Against the First Best}\label{sec:first-best}
In this section, we analyze the performance of the best contract against the first best benchmark of social welfare when ignoring incentive constraints, i.e., $\opt = \max_{S} (f(S) -c(S))$.

We first show that when there is only one agent with $m$ actions and a subadditive reward function, then the principal's utility of the optimal contract is at least  $\frac{\opt}{m}$.

\begin{proposition} \label{prop:welfare-single}
For every subadditive reward function $f:2^\actions \rightarrow \reals_{\geq 0}$ over $|A|=m$ actions and a single agent with costs $c_1,\ldots,c_m$, there exists a contract $\alpha$ that guarantees a utility for the principal of at least  $\frac{\opt}{m}$, where $\opt=\max_{S} (f(S)-c(S))$.
\end{proposition}
\begin{proof}
For this proof, we consider the definition of critical points defined in \cite{DuttingEFK21}.
A critical point is a value of $\alpha \in (0,1)$ for which the best responses to contract $\alpha +\epsilon$ are different than the best responses to contract $\alpha -\epsilon$ for every $\epsilon>0$.
In \cite{DuttingEFK21}, it was shown that this set is always finite. Furthermore, if we order these critical points and denote them by $(0=\alpha_0<\alpha_1,...< \alpha_{k}\leq 1)$, then there exists a sequence of action sets  
$S_0 , S_1, ... ,S_k$ such that for each $i =1,\ldots k $,  the sets $S_i,S_{i-1}$ are best responses to contract $\alpha_i$, the set $S_0$ is the best response contract $\alpha_0=0$ that maximizes the principal utility (i.e., $S_0 = \{ j\in\actions \mid c_j=0\}$). Moreover, the set $S_k$ is the best response under contract $\alpha=1$, it holds that $f(S_0) < f(S_1) <\ldots < f(S_k)$, and for every $i=1,\ldots,k$ it holds that  \begin{equation}
  \alpha_i = \frac{c(S_i) - c(S_{i-1})}{f(S_i) - f(S_{i-1})} . \label{eq:alpha-i}
\end{equation} 

Based on these observations, defining $\alpha_{k+1} = 1$, the social welfare can be rewritten by
\begin{eqnarray}    
\opt  & =& f(S_k) -c(S_k) = f(S_0) + \sum_{i=1}^k f(S_i) - f(S_{i-1}) - (c(S_i)- c(S_{i-1})) \nonumber \\ & \stackrel{\eqref{eq:alpha-i}}{=} & f(S_0) + \sum_{i=1}^k f(S_i) - f(S_{i-1}) - \alpha_i(f(S_i)- f(S_{i-1})) = f(S_0) + \sum_{i=1}^k (1-\alpha_i)(f(S_i) - f(S_{i-1})) \nonumber \\ &= & f(S_0) + \sum_{i=1}^k (\sum_{t=i}^k \alpha_{t+1}-\alpha_t)(f(S_i) - f(S_{i-1}))  = f(S_0) + \sum_{t=1}^k\sum_{i=1}^t ( \alpha_{t+1}-\alpha_t)(f(S_i) - f(S_{i-1}))      \nonumber \\ & = & f(S_0) + \sum_{t=1}^k( \alpha_{t+1}-\alpha_t)(f(S_t )-f(S_0))  =  \sum_{t=0}^k( \alpha_{t+1}-\alpha_t)f(S_t ),   \nonumber
\end{eqnarray}
where the first equality is since $S_k$ is the best response under contract $\alpha=1$, which by definition maximizes the social welfare; the second equality is by telescoping and since $c(S_0) =0$.

Let $\ell_1=0$ if $S_0 \neq \emptyset$, otherwise let $\ell_1 = 1$. So $S_{\ell_1}$ is the first non-empty set. For $i>1$ let $\ell_i = \min\{ \ell > \ell_{i-1} \mid   S_\ell  \not\subseteq \cup_{j=1}^{\ell_{i-1}}   S_j \}$ (if no such $\ell$ exists, then let $\ell_i=\infty$). That is, $\ell_i$ is the smallest $\ell \geq \ell_{i-1}$ such that $S_\ell$ contains an action that is not contained in any $S_1, \ldots, S_{\ell_{i-1}}$. Let $d$ be the last index for which $\ell_d \neq \infty$, and let $T_i = \cup_{j=1}^{\ell_i} S_j$. Note that $d \leq m$.
For convenience, let $\alpha_{\ell_{d+1}} =1$, and let $T_0 = \emptyset$. Then, we can bound the expression for social welfare derived above by
\begin{eqnarray}
\opt = \sum_{t=0}^k( \alpha_{t+1}-\alpha_t)f(S_t ) & = & \sum_{i=1}^d\sum_{t=\ell_i}^{\ell_{i+1}-1}( \alpha_{t+1}-\alpha_t)f(S_t ) 
\leq 
\sum_{i=1}^{d} (\alpha_{\ell_{i+1}}-\alpha_{\ell_i}) f(T_i)  \nonumber \\
& \leq &  \sum_{i=1}^{d} (\alpha_{\ell_{i+1}}-\alpha_{\ell_i}) \sum_{j=1}^i f(T_j \setminus T_{j-1})   = \sum_{j=1}^d f(T_j \setminus T_{j-1}) \sum_{i=j}^{d} (\alpha_{\ell_{i+1}}-\alpha_{\ell_i}) \nonumber \\
& = & \sum_{j=1}^d f(T_j \setminus T_{j-1})  (1-\alpha_{\ell_j}), \label{eq:bound-sw}
\end{eqnarray}
where the first inequality is since for every $t < \ell_{i+1}$, $S_t \subseteq T_i$; and the second inequality is by subadditivity.
 
On the other hand, since $ \cup_{j=1}^{\ell_i} S_j =T_i$, and $\cup_{j=1}^{\ell_i-1} S_j =\cup_{j=1}^{\ell_{i-1}} S_j =T_{i-1}$
we get that $T_i \setminus T_{i-1} \subseteq S_{\ell_i}$, which by combining with 
monotonicity of $f$ it implies that $f(S_{\ell_i}) \geq f(T_i \setminus T_{i-1})$.

So, in total, the best contract among $\alpha_{\ell_1}, \alpha_{\ell_d}$ gives a principal utility of
\[
\max_{i \in [d]} (1-\alpha_{\ell_i})f(S_{\ell_i}) \geq \frac{1}{d} \sum_{i = 1}^d (1-\alpha_{\ell_i})  f(T_i \setminus T_{i-1}) \geq \frac{1}{d} \opt.
\]
Since $d \leq m$, this shows the claim.
\end{proof}

Notably, Proposition~\ref{prop:welfare-single} is tight. 
Indeed,
it was shown in \cite{DuttingRT19} that even in the non-combinatorial model (which corresponds to $f$ being of the form $f(S) = \max_{j \in S} f(j)$), the gap between the social welfare and the maximum principal's utility can be at least $m$.

We next show how to generalize Proposition~\ref{prop:welfare-single} to the case of multiple agents.
\begin{theorem}\label{thm:gap-multi-best}
        For every subadditive reward function $f:2^\actions \rightarrow \reals_{\geq 0}$ over $|A|=m$ actions and multiple agents with action costs $c_1,\ldots,c_m$, there exists a contract $\contract$ with a corresponding equilibrium in which the utility for the principal is at least $\frac{\opt}{m}$.
\end{theorem}
\begin{proof}
    For every agent $i$, let $m_i = |A_i|$.
    Let $S^\star$ be the set that maximizes the social welfare.
    By Proposition~\ref{prop:welfare-single} we know that for every agent $i$ exists contract $\alpha_i$ that guarantees for the principal a utility of at least $\frac{f(S^\star_i)-c(S^*_i)}{m_i}$ assuming that all the other agents do nothing.
    Consider a distribution over contracts, which chooses agent $i$ with probability $\frac{m_i}{m}$, offering him a contract $\alpha_i$ and all other agents receive $0$.
    Under the best equilibrium of each such contract, we get that the expected utility of the principal is at least \begin{equation} 
     \sum_i \frac{m_i}{m}\frac{f(S^\star_i)-c(S^*_i)}{m_i} = \sum_i \frac{f(S^\star_i)-c(S^*_i)}{m}  \geq \frac{f(S^\star)-c(S^\star)}{m} = \frac{\opt}{m}. \label{eq:welfare}
     \end{equation}
    Thus, there exists $i^\star$ such that the best contract (where only agent $i^\star$ takes actions and receive payments) achieves a principal's utility of at least $\frac{\opt}{m}$, which concludes the proof.
\end{proof}

Recall that if $f$ is submodular our algorithm ensures (as stated in Theorem~\ref{theorem:nolargeorsingleagent}) that \emph{every} equilibrium of the computed contract is a constant-factor approximation of the \emph{best} equilibrium of the \emph{best} contract. As an immediate corollary of Theorem~\ref{thm:gap-multi-best}, we therefore obtain the following bound that also lower-bounds the worst equilibrium.
\begin{corollary}
        \label{thm:gap-multi-submod}
        For every submodular reward function $f:2^\actions \rightarrow \reals_{\geq 0}$ over $|A|=m$ actions and multiple agents with action costs $c_1,\ldots,c_m$,  there exists a contract $\contract$ for which every corresponding equilibrium guarantees a utility for the principal of at least $\Omega(\frac{\opt}{m})$.
\end{corollary}

For completeness, we next show that for the multiple agent case, the gap between the social welfare and the optimal principal's utility can be as large as $\Omega(n)$ even in the binary action model of \cite{DuettingEFK23} and even under additive reward functions. It was shown in \cite{cacciamani2024multi} that this gap is generally unbounded when the reward function is supermodular.  
\begin{example}
    Consider a setting with $n$ agents with binary actions where the reward function  $f:2^{\actions} \rightarrow \reals_{\geq 0} $ is $f(S) =\frac{|S|}{n}$ while the cost for every agent is $\frac{1}{2n}$.
    The social welfare of this instance is when all agents exert effort leading to a welfare of $\frac{n}{n}-n\cdot \frac{1}{2n} =\frac{1}{2}$. However, no contract can guarantee the principal a utility of more than $\frac{1}{2n}$ as the principal needs to pay each agent at least $\alpha_i \geq \frac{1}{2}$ for him to exert effort, thus the principal cannot incentivize more than one agent while having a positive utility. The optimal contract that incentivizes at most one agent achieves a principal's utility of  $\frac{1}{2n}$.
\end{example}

If we assume $f$ only be subadditive but not necessarily submodular, the gap to the welfare under worst-case equilibrium can indeed be larger. We first give a positive result, showing that the gap is no larger than $O(n \cdot m)$.
\begin{proposition}
        For every subadditive reward function $f:2^\actions \rightarrow \reals_{\geq 0}$ over $|A|=m$ actions and $n$ agents with action costs $c_1,\ldots,c_m$, there exists a contract $\alpha$ for which every corresponding equilibrium guarantees a utility for the principal of at least $\Omega(\frac{\opt}{n\cdot m})$.
\end{proposition}
\begin{proof}
    Let $A^0$ be the set of actions with a cost of $0$, let $A^0_i = A_i \cap A^0$, let $m_i = |A_i|$, let $i_0^\star = \arg\max_{i \in \agents} f(A_i^0)$, let $\opt_i = \max \{f(S_i) -c(S_i) \mid S_i \subseteq A_i\}$,  and let $i^\star = \arg\max_{i \in \agents} \frac{\opt_i}{m_i}$. 
    We know by Equation~\eqref{eq:welfare} that $\opt_{i^\star} \geq \frac{\opt}{m}$.
    We consider two cases:
    If $f(A_{i^\star_0}^0) \geq \frac{\opt}{4n\cdot m}$, then we analyze the equilibria of contract $\contract$ where agent $i^\star_0$ is offered $\epsilon >0  $, and all other agents are offered $0$.
    Let $S$ be some equilibrium with respect to $\contract$.
    Assume towards contradiction that $f(S) < f(A_{i^\star_0}^0)$.  Then agent $i^\star_0$ can deviate to $A_{i^\star_0}^0$  which will strictly increase his utility since $$\epsilon \cdot f(S_{-i^\star_0} \cup A_{i^\star_0}^0 ) -c(A_{i^\star_0}^0) \geq \epsilon \cdot f( A_{i^\star_0}^0 )  > \epsilon \cdot f(S)  \geq \epsilon \cdot f(S) -c(S_i) . $$
    So in this case, the principal's utility under equilibrium $S$, is at least $ (1-\epsilon) f(A_{i^\star_0}^0) = \Omega(\frac{\opt}{n\cdot m}) $.

    Now consider the case that $f(A_{i^\star_0}^0) < \frac{\opt}{4n\cdot m}$. We know by Proposition~\ref{prop:welfare-single} that if assuming that all agents but $i^\star$ do nothing, then there exists a contract $\alpha_{i^\star}^\star$ with a corresponding best response $S^\star_{i^\star}$ that gives the principal a utility of at least $\frac{\opt_{i^\star}}{m_{i^\star}}$. Consider a contract $\contract$ where $\alpha_{i^\star} = \frac{1+\alpha^\star_{i^\star}}{2}$, and   $\alpha_i = 0$ for $i \neq i^\star$. Let $S$ be an equilibrium with respect to $\contract$. 
    Assume towards contradiction that $f(S) < \frac{f(S_{i^\star}^\star)}{2}$ then we know that 
    \begin{equation}
        \alpha_{i^\star}^\star f(S_{i^\star}^\star) -c(S_{i^\star}^\star) \geq \alpha_{i^\star}^\star f(S_{i^\star}) -c(S_{i^\star}). \label{eq:alpha-star}
    \end{equation} 
    On the other  hand, the utility of agent $i^\star$ satisfies 
    \begin{eqnarray}
         \alpha_{i^\star} f(S) -c(S_{i^\star})  & \leq  & \alpha_{i^\star} (f(S_{i^\star}) + f(S_{-i^\star})) -c(S_{i^\star})  \nonumber \\ & \leq & f(S_{-i^\star}) +  \alpha_{i^\star} f(S_{i^\star})  -c(S_{i^\star}) \nonumber \\ 
         & =  &f(S_{-i^\star}) + \frac{1-\alpha_{i^\star}^\star}{2}f(S_{i^\star})  +  \alpha_{i^\star}^\star f(S_{i^\star})  -c(S_{i^\star}) \nonumber \\ 
         & \stackrel{\eqref{eq:alpha-star}}{\leq} &  f(S_{-i^\star}) + \frac{1-\alpha_{i^\star}^\star}{2}f(S_{i^\star})  +  \alpha_{i^\star}^\star f(S_{i^\star}^\star) -c(S_{i^\star}^\star) \nonumber \\ 
         & \leq & f(A^0) +  \frac{1-\alpha_{i^\star}^\star}{2}f(S)  +  \alpha_{i^\star}^\star f(S_{i^\star}^\star) -c(S_{i^\star}^\star) \nonumber \\
& < & n \cdot f(A_{i_0^\star}^0) +  \frac{1-\alpha_{i^\star}^\star}{2}\frac{f(S_{i^\star}^\star)}{2}  +  \alpha_{i^\star}^\star f(S_{i^\star}^\star) -c(S_{i^\star}^\star) \nonumber \\
& < &
n \cdot \frac{\opt}{4n\cdot m}+  \frac{1-\alpha_{i^\star}^\star}{2}\frac{f(S_{i^\star}^\star)}{2}  +  \alpha_{i^\star}^\star f(S_{i^\star}^\star) -c(S_{i^\star}^\star) \nonumber \\
& \leq &
 \frac{(1-\alpha_{i^\star}^\star)f(S_{i^\star}^\star)}{4}+  \frac{1-\alpha_{i^\star}^\star}{2}\frac{f(S_{i^\star}^\star)}{2}  +  \alpha_{i^\star}^\star f(S_{i^\star}^\star) -c(S_{i^\star}^\star) \nonumber \\ 
 & = &    \alpha_{i^\star} f(S_{i^\star}^\star) -c(S_{i^\star}^\star), \nonumber
    \end{eqnarray}
    where the first inequality is by subadditivity; the second inequality is since  $\alpha_{i^\star} \leq 1$; the fourth inequality is since agents that are not $i^\star$ don't take costly actions since their contract is $0$ and by monotonicity of $f$; the fifth inequality is by subadditivity of $f$, and by our assumption towards contradiction; the sixth inequality holds since we consider the case where $f(A_{i^\star_0}^0) < \frac{\opt}{4n\cdot m}$; the seventh inequality is  by Proposition~\ref{prop:welfare-single}.

    Thus, it holds that $f(S) \geq \frac{f(S_{i^\star}^\star)}{2}$, and the principal's utility is at least $$ (1-\alpha_{i^\star}) f(S) \geq \frac{1-\alpha_{i^\star}^\star}{2} \frac{f(S_{i^\star}^\star)}{2} = \Omega \left(\frac{\opt}{n\cdot m}\right) , $$
    which concludes the proof.
\end{proof}

Finally, we show that for subadditive reward functions (or even XOS) the gap between the welfare and the worst case equilibrium under any contract can be quadratic.
\begin{example}\label{ex:xos}
Consider an instance with $n=2k$ agents denoted by $\agents$, and binary actions.
The reward function $f:2^\actions \rightarrow \reals_{\geq 0}$ is $$ f(S) = \frac{\max\{|S|, k^2 \cdot |S \setminus [k]|, 2 \cdot \ind{S \neq \emptyset}\}}{k^3}.$$

The function is XOS since it is the maximum over the additive functions $\frac{|S |}{k^3}$, $  \frac{k^2 \cdot|S \setminus [k]|}{k^3}$, and the function $g_i(S) = \frac{2 \cdot \ind{i \in S}}{k^3}$ for every $i$.
The costs of exerting effort for agents in $[k]$ are $0$, and the costs of agents in $[n]\setminus [k]$ are $\frac{k^2-k/2}{k^3}$.

The social welfare of this instance is when all agents exert effort, leading to a social welfare of $$\opt = 1 - k \cdot \frac{k^2-k/2}{k^3} = \frac{1}{2k}.$$
On the other hand, let $\contract$ be some contract.
If for all agents $i\in\agents \setminus [k]$ it holds that $\alpha_i <\frac{2k-1}{2k}$, then the set  $\{1\}$ is an equilibrium. To see this, observe that agent $i \in \actions\setminus [ k]$ does not want to deviate since 
$$ \alpha_i\cdot f(\{i,1\}) -c(i) \leq \frac{2k-1}{2k} \frac{1}{k^2} - \frac{k^2-k/2}{k^3} = 0 \leq \alpha_i \cdot f(1).$$
Agent $i \in [k] \setminus \{1\}$ does not want to deviate since $$ \alpha_i \cdot  f(\{i,1\}) -c(i) = \alpha_i  \cdot f(1) .$$
Thus, in this case, there is an equilibrium with a principal's utility of at most $f(1) = \frac{2}{k^3}=O\left(\frac{\opt}{n^2}\right)$.

If exists an agent $i\in\agents \setminus [k]$ with $\alpha_i \geq\frac{2k-1}{2k}$ then, there must be at most one such agent.
Now consider the equilibrium where all agents in $[k]$ exert effort. Then agent $i$ (and also all agents in $\agents\setminus [k]$) don't want to exert effort since $$ \alpha_{i} f([k]) = \alpha_i \cdot \frac{1}{2k^2}  > \alpha_i \cdot \frac{1}{k}-\frac{k^2-k/2}{k^3}= \alpha_i f([k] \cup \{i\}) -c(i) ,$$
where the inequality is since $\alpha_i\leq 1$.
Now, the utility of the principal under this equilibrium is at most $$(1-\alpha_i) f([k]) \leq \frac{1}{2k}  \cdot \frac{1}{k^2} = O\left(\frac{\opt}{n^2}\right).$$
\end{example}

\bibliographystyle{plain}
\bibliography{references}

\appendix
\section{Proof of Proposition~\ref{prop:potential-function}}\label{apx:potential-function}

Fix $\contract \in [0,1]^n$. Consider $\phi(S): 2^\actions \rightarrow \reals$, with $\phi(S) = f(S) - \sum_{i \in \agents} \sum_{j \in S_i} \frac{c_j}{\alpha_i}$. If $\alpha_i = 0$, we define $c_j/\alpha_i = \infty$ when $c_j > 0$, and we let $c_j/\alpha_i = 0$ when $c_j = 0$. 

Note that when $\alpha_i = 0$, any set of actions $S_i \subseteq \actions_i$ that contains an item $j \in S_j$ such that $c_j > 0$ is strictly dominated (e.g., by the empty set $\emptyset$). We thus assume that for agents $i$ such that $\alpha_i = 0$ all actions in $\actions_i$ have cost $c_j = 0$. Note that this in particular means that agent $i$ receives a utility of zero for all possible sets of actions $S_i \subseteq \actions_i$.

This is a somewhat subtle point, because although the agent with $\alpha_i = 0$ is indifferent which set of actions with zero cost he chooses, the principal derives different rewards from different choices; and also the strategic choices of the other agents are influenced (through the payments they receive which depend on the principal's reward).

\begin{claim}\label{cla:gen-ord-pot}
Function $\phi: 2^\actions \rightarrow \reals$ is a generalized ordinal potential function. 
\end{claim}
\begin{proof}
Recall that agent $i$'s utility is $u_i(S,\contract) = \alpha_i f(S) - \sum_{j \in S_i} c_j$. We need to show that for any $i, S_i,S'_i,S_{-i}$
\[
u_i(S_i,S_{-i},\contract) > u_i(S'_i,S_{-i},\contract) \implies \phi(S_i,S_{-i}) > \phi(S'_i,S_{-i}).
\]

We only need to consider $i$ such that $\alpha_i > 0$. Note that, since $\alpha_i > 0$, if $u_i(S_i,S_{-i},\contract)  > u_i(S'_i,S_{-i},\contract)$, then 
\[
f(S_i,S_{-i}) - \sum_{j \in S_i} \frac{c_j}{\alpha_i} > f(S'_i,S_{-i}) - \sum_{j \in S'_i} \frac{c_j}{ \alpha_i}.
\]
Thus, $u_i(S_i,S_{-i},\contract) > u_i(S'_i,S_{-i},\contract)$ implies that
\[
\phi(S_i,S_{-i}) - \phi(S'_i,S_{-i}) = \left(f(S_i,S_{-i}) - \sum_{j \in S_i} \frac{c_j}{\alpha_i}\right) - \left(f(S'_i,S_{-i}) - \sum_{j \in S'_i} \frac{c_j}{\alpha_i}\right) > 0,
\]
as claimed.
\end{proof}

Proposition~\ref{prop:potential-function} follows from Lemma 2.5 of \cite{MondererS96} which shows that any finite game that admits a generalized ordinal potential function has the finite improvement property, and hence at least one (pure) Nash equilibrium.

\section{Demand Queries and Equilibria}

Fix $\contract \in [0,1]^n$. Recall the definition of the potential function $\phi(S): 2^\actions \rightarrow \reals$, with $\phi(S) = f(S) - \sum_{i \in \agents} \sum_{j \in S_i} \frac{c_j}{\alpha_i}$. Consider a demand query with respect to $f$ at prices $p_j = c_j/\alpha_i$ for $j \in \actions_i$ such that $\alpha_i > 0$, and $p_j = 0$ for $j \in \actions_i$ such that $c_i = 0$. 

As discussed in Appendix~\ref{apx:potential-function}, We can assume that for agents such that $\alpha_i = 0$, the action sets $\actions_i$ are pruned so that $c_j = 0$ for all $j \in \actions_i.$

\begin{claim}
Any $S \subseteq \actions$ that maximizes $\phi(S)$ is a (pure) Nash equilibrium. 
\end{claim}
\begin{proof}
Note that if $S \subseteq \actions$ maximizes $\phi(S)$, then since $\phi$ is a generalized ordinal potential function (Claim~\ref{cla:gen-ord-pot}), if for some agent $i$ there existed a $S'_i$ such that $u_i(S'_i, S_{-i},\contract) > u_i(S_i, S_{-i},\contract)$, then we would also have $\phi(S'_i,S_{-i}) > \phi(S_i,S_{-i})$, in contradiction to $S$ being a maximum of $\phi$. So  , for all $i$ and all $S'_i$, we must have $u_i(S_i,S_{-i},\contract) \geq u_i(S'_i,S_{-i},\contract)$.

Note that, in particular, since for every $i$ and any $S'_{-i}$, it holds that $u_i(\emptyset,S'_{-i},\contract) = \alpha_i f(\emptyset,S'_{-i}) \geq 0$, we must have $u_i(S_i,S'_{-i}) \geq 0$.
\end{proof}

More generally, there is a one-to-one correspondence between equilibria and local optima of the potential function, where ``local'' refers to any change in actions controlled by a single agent.

We next show that when the reward function is XOS, then there is an instance in which for every contract, there is a bad equilibrium. This is true even in the case of binary action for every agent.
\begin{proposition} \label{prop:bad-xos}
There exists an instance with XOS reward function $f:2^\actions\rightarrow \reals_{\geq 0}$  and costs $c$, for which for every contract $\contract$ there exists an equilibrium $S$ such that the principal's utility is at most $O\left(\frac{u_p(S^\star, \contract^\star)}{|\actions|}\right)$,
where $\contract^\star,S^\star$ are the best contract with the best corresponding equilibrium.
\end{proposition}
\begin{proof}
    Consider and instance with $n>2$ agents with binary actions where the reward function $f:2^\actions\rightarrow \reals_{\geq 0}$ is $$ f(S) = \begin{cases} \frac{2}{n} & \text{ if $\lvert S \rvert = 1$} \\ 
    \frac{\lvert S \rvert}{n} & \text{ otherwise}\end{cases} $$
    and let the cost of all agents be $c_i= \frac{1}{2\cdot n^2}$. The reward function is XOS since it is the maximum over the additive function $g(S) =\frac{|S|}{n}$ and the functions $g_i(S) = \frac{2 \cdot \ind{ i \in S}}{n}$ for every agent $i$.
    The optimal contract $\contract^\star$ is to pay each agent $\alpha_i=\frac{1}{2n}$ of the reward for which the optimal equilibrium is $S^\star =\actions$, which leads to a principal utility of $u_p(S^\star, \contract^\star) = (1-n\cdot \frac{1}{2n}) \cdot \frac{n+0}{n} = \frac{1}{2} $.
    However, for every contract $\contract$ there exists an equilibrium in which at most one agent exerts effort.
        
    To see this,  let $S$ be a set of agents of size at least two which is an equilibrium with respect to $\contract$. (We know that an equilibrium exists, and if all equilibria are at most of size one, then we are done.) 
    It means that for every $i \in S$ it holds that $\alpha_i f( i \mid S \setminus \{i\}) \geq c_i$ as otherwise the agent would prefer to not exert effort. By rearranging we get that $\alpha_i \geq \frac{c_i}{f( i \mid S \setminus \{i\})} \geq \frac{\nicefrac{1}{2n^2}}{\nicefrac{1}{n}} = \frac{1}{2n}$. Let $i^\star$ be an arbitrary agent in $S$ (such exists since $|S| \geq 2$).
    We claim that the set $S' = \{i^\star\}$ is also an equilibrium with respect to $\contract$.
    This is true since for agent $i^\star$ it holds that $\alpha_{i^\star} f( i^\star \mid S' \setminus \{i^\star\}) = \alpha_{i^\star} \frac{2}{n} \geq  \frac{1}{2n^2} =c_{i^\star}$.
    For every agent $i \neq i^\star$ it holds that $f(\{i, i^\star\}) = \frac{2}{n} = f(\{i^\star\})$ and so it holds that  $\alpha_{i} f( i \mid S' \setminus \{i\}) = \alpha_{i^\star} 0 \leq  \frac{1}{2n^2} =c_{i^\star}$.

    Thus, any contract has an equilibrium that incentivizes at most one agent, which leads to a utility of at most $\frac{2}{n}$.
\end{proof}

\begin{proposition}\label{prop:sup-max}
    For every instance defined by $f:2^\actions\rightarrow \reals_{\geq 0},c$, the supremum is attained at some contract $\contract^\star$ with a corresponding set $S^\star$.
\end{proposition}
\begin{proof}
    For every set of actions $S$, the set of contracts that incenticizes it $T_S =\{ \contract \mid S \in \nash{\contract}\}$ is the set of actions that satisifirs for every $i$ and $S_i'\subseteq A_i$ that $\alpha_i f(S) -c(S_i) \geq \alpha_i f(S_i',S_{-i}) -c(S_i')$. 
    Since when fixing the actions of the agents, the principal's utility is a continuous function, and the set   $T_S$ is compact, we get by the Extreme Value Theorem that there exists a contract $\contract_S \in T_S$ for which   
    $$u_p(S, \contract_S)=\sup_{\contract}u_p(S, \contract).$$
    Let $S^\star = \arg\max u_p(S, \contract_S)$.
    By definition, $S^\star \in \nash{\contract_{S^\star}}$.
    Then, for every $\alpha$, and $S \in \nash{\contract}$ it holds that 
    $$u_p(S, \contract) \leq  u_p(S, \contract_S) \leq u_p(S^\star, \contract_{S^\star}),$$
    which concludes the proof.
\end{proof}

\end{document}